\documentclass[twocolumn,final]{IEEEtran}

\usepackage{amsmath}
\usepackage{graphicx}
\usepackage{epsfig}
\usepackage{array}
\usepackage{psfrag}
\usepackage{amssymb}
\usepackage{color}
\usepackage[active]{srcltx}
\usepackage{amsthm}
\usepackage{varwidth}
\usepackage[ruled,linesnumbered]{algorithm2e}
\usepackage{algorithmicx}
\usepackage{epstopdf}
\newcommand{\bff}{\mathbf f}
\newcommand{\bg}{\mathbf g}

\newcommand{\by}{\mathbf y}

\newcommand{\bw}{\mathbf w}
\newcommand{\bn}{\mathbf n}

\newcommand{\bt}{\mathbf t}

\newcommand{\bx}{\mathbf x}

\newcommand{\bs}{\mathbf s}

\newcommand{\bD}{\mathbf{D}}
\newcommand{\bE}{\mathbf{E}}

\newcommand{\bG}{\mathbf{G}}
\newcommand{\bH}{\mathbf{H}}

\newcommand{\bW}{\mathbf{W}}
\newcommand{\bX}{\mathbf{X}}
\newcommand{\bY}{\mathbf{Y}}

\makeatletter
\newcommand{\removelatexerror}{\let\@latex@error\@gobble}
\makeatother

\author{Adrian\ Schad \quad \quad Ka L. Law \quad \quad Marius Pesavento
}
\title{Rank-Two Beamforming and Power Allocation in Multicasting Relay Networks}

\begin{document}

\maketitle

\begin{abstract}
In this paper, we propose a novel single-group multicasting relay beamforming
scheme. We
assume a
source that transmits common messages via multiple
amplify-and-forward relays to multiple destinations. To increase the
number of
degrees of freedom in the beamforming design, the relays process two received
signals
jointly and transmit the
Alamouti
space-time block code over two different beams. Furthermore, in  contrast to the
existing
relay
multicasting scheme of the literature, we take into account the direct links
from the source to the destinations. We aim to
maximize the lowest received quality-of-service by choosing the proper relay
weights and the ideal distribution of the power resources in the network.  To
solve the corresponding
optimization problem, we propose an iterative algorithm which solves
sequences of convex approximations of the original non-convex optimization
problem.
Simulation results
demonstrate significant performance improvements of the proposed methods as
compared with
 the existing relay
multicasting scheme of the literature and an algorithm based on the popular
semidefinite relaxation technique.

\end{abstract}


\section{INTRODUCTION}
``Beamforming is a versatile and powerful approach to receive,
transmit, or relay signals-of-interest in a spatially selective
way in the presence of interference and noise'' \cite{Exle}.
Recently, the concept of  receive and transmit
beamforming has been used to enhance coverage and data rate performance in
amplify-and-forward (AF) relay networks. Distributed beamforming has been
applied as
a coherent transceiver technique in a variety of relay network architectures,
such as single-user networks \cite{jafar_netw_subm}--\cite{Chen08}. In such
networks one
source transmits data to one destination. Single-user networks have been
extended to peer-to-peer networks, where multiple
source-destination pairs communicate directly via relays
\cite{fazeli08}, \cite{MultiUserTwoWay}, \cite{Schad2}.
Moreover, distributed systems of non-connected relays can be used for
multicasting to
transmit data from one source to many destinations simultaneously to
avoid exhaustive individual transmissions \cite{Bornhorst}--\cite{12}.

The task to select the optimum antenna
weights in centralized beamforming systems using a connected antenna array
\cite{TransmitBeamforming}--\cite{ERt}, or in distributed beamforming
systems using a non-connected array
\cite{Bornhorst}--\cite{12} for multicasting is
highly non-trivial as it requires to form beams towards several destinations,
each corresponding to a different spatial signature. In
\cite{TransmitBeamforming}, it has
been shown that the problem of selecting the antenna weights for single-group
multicasting is NP-hard. Until now, no polynomial time algorithm for NP-hard
problems
is known and it is expected that exact solutions can only be computed
within exponential time.

To derive adequate solutions to beamforming
problems which belong to the class of
non-convex quadratically constrained quadratic optimization problems (QCQPs),
computationally efficient algorithms have been proposed which approximate the
feasible set of the optimization problem
\cite{MultiUserTwoWay}--\cite{Eusipco},
\cite{Schad2}--\cite{9}, \cite{ERt}--\cite{Dartman}.

Outer approximation techniques replace the QCQPs
by convex semidefinite programs (SDPs) which can be solved efficiently
\cite{Schad2}--\cite{9}.
In the latter approach, the weight vector is replaced by a positive
semidefinite Hermitian matrix. Since
this so-called semidefinite relaxation (SDR) technique extends the feasible
region, a solution matrix to the SDR problem lies not necessarily in the
feasible set of the original problem. If the rank
$\mathcal{R}(\mathbf{X}^\star)$ of the solution
matrix $\mathbf{X}^\star$ equals to
one,
the SDR solution is a global solution to the original problem. In
practice, however, $\mathcal{R}(\mathbf{X}^\star)$ might be greater than one
and $\mathbf{X}^\star$ is not feasible for the original problem. Especially in
single-group multicasting
scenarios, where many destinations demand a  minimum received
quality-of-service (QoS), it is not likely that
$\mathcal{R}(\mathbf{X}^\star)=1$. This is a
consequence of the fact that the
existence of an SDR solution matrix with rank $\mathcal{R}(\mathbf{X}^\star)$ is
only
guaranteed if $\mathcal{R}(\mathbf{X}^\star)\geq \mathcal{O}(\sqrt{M})$, where
$M$ is the number
of destinations \cite{SDPRelax}. If $\mathcal{R}(\mathbf{X}^\star)>1$, the
objective value of the SDR solution is only a lower bound to the objective value
of the QCQP. The accuracy of the lower bound decreases with a growing
number of destinations $M$ \cite{Rank2}. Therefore, for multicasting with large
$M$, the lower
bound
generated by SDR can lie quite far from the true minimum value for rank-one
beamforming.

\textcolor{black}{Interestingly, the formulations of the optimization problems for single-group multicasting are similar to the multi-user downlink beamforming problems of \cite{MUD1}--\cite{MUD3}. The latter works consider scenarios where each destination receives individual signals and exploit
uplink-downlink duality to derive algorithms that solve joint beamforming and power allocation problems. However, for single-group multicasting problems, these algorithms are not applicable due to the duality gap \cite{TransmitBeamforming}.}

Recently, in the two independent works \cite{Rank2} and \cite{Xin},
rank-two transmit beamforming techniques for multicasting networks have been
proposed in
which also {\it rank-two} SDR solution matrices are feasible.
In
these techniques, two weight vectors are used at the transmitter to process
two data symbols jointly. Rank-two beamforming techniques have gained much
interest in the current research as the system performance is enhanced due to
the
increased number of degrees of freedom resulting from the additional weight
vector \cite{Eusipco}, \cite{Rank2}--\cite{9}. \textcolor{black}{ The gain in performance comes at
virtually no additional cost of decoding at
the receiver and symbol-by-symbol detection can be applied by utilizing Alamouti's orthogonal space-time block coding (OSTBC)
\cite{Alamouti}. Note that Alamouti's code has been further developed in \cite{Tarokh}--\cite{Jafarkhani5}. Moreover, it has been combined with a variety of signal processing techniques for multi-antenna systems, as for instance,  multi-user detection \cite{MUD}, \cite{Reynolds}, transmit beamforming  with limited channel feedback \cite{Zhou1}, \cite{Zhou2}, receive beamforming \cite{Sun}, and interference alignment \cite{li11}, \cite{Zaki}.}

In this paper, we propose a distributed rank-two beamforming scheme for single-group multicasting using a
network of amplify-and-forward (AF) relays. In
AF multicasting
networks, the relays forward common messages from a single source to
multiple destinations. The proposed AF single-group multicasting
scheme (AFMS) is a
non-trivial extension of the
transmit beamforming technique
of \cite{Rank2} and \cite{Xin} to a distributed beamforming system. \textcolor{black}{ In such a system, the retransmission of the noise at the relays  generally leads to correlated noise at the destinations, even when OSTBCs are used \cite{DistributedSpaceTime1}. Then, symbol-by-symbol detection is not optimal. Here, we design a distributed beamforming scheme which achieves uncorrelated destination noise.} We
refer
to our scheme as the Rank-2-AFMS in
distinction to the conventional Rank-1-AFMS of \cite{Bornhorst}.
As another generalization to the Rank-1-AFMS,
we
exploit direct link connections from the source to the destinations.

As the design criterion to select the power at the source and the relay weights,
we
aim to maximize the minimum
QoS at the destinations under constraints on the
transmit power in the network.  We consider
constraints on the maximum transmit power of the source, on the individual
power of every relay, on the sum power of the relays, and on the total power by
both the relays and source.

To solve the non-convex max-min fairness optimization problem of jointly
determining
the relay weight vector and the power split between the relays and the source,
we
propose a linearization-based iterative algorithm. \textcolor{black}{ This algorithm belongs
to the class
of concave-convex procedure (CCCP) algorithms \cite{CCCP}.} In advantage to the SDR
technique, where the number of variables is
roughly squared, the number of
variables is not increased for CCCP algorithms. CCCP
algorithms are used to approximately solve non-convex
{ difference of convex (DC)
programming problems.} Many optimization problems that arise in the context of
wireless communications are DC problems, including power allocation
\cite{Phan}--\cite{Yong} and beamforming problems \cite{1}--\cite{30}.
In contrast to the algorithms for max-min fair beamforming of
\cite{MultiUserTwoWay}, \cite{Schad2},
\cite{multigroupMulticast}, \cite{30}, and \cite{Dartman2}, which treat solely
the optimization of beamforming vectors, our
algorithm derives the relay weight vectors and allocates transmit power to
the source and the relays. The max-min fair beamforming
approaches in \cite{Schad2} and \cite{multigroupMulticast} combine the SDR
technique with one-dimensional (1D) search on the maximum QoS. To compare our
CCCP algorithm (Max-Min-CCCP) with the latter SDR-based
algorithms, we combine
the SDR technique with two-dimensional (2D) search on both the maximum QoS and the best
power split between the relays and the source.

To test the Rank-2-AFMS and the Rank-1-AFMS under realistic conditions, we
use the channel model of  \cite{Baum}. The simulation results
demonstrate the performance of the proposed rank-two
scheme combined with the proposed algorithm compared with the rank-one scheme of
 \cite{Bornhorst} and the theoretical bound obtained by SDR. The
Max-Min-CCCP
algorithm outperforms the SDR technique for high destination numbers at a much
lower runtime. Moreover, the Max-Min-CCCP
algorithm offers a good performance-runtime trade-off and achieves a minimum
signal-to-noise ratio (SNR) which is less
than 1 dB lower than the theoretical bound  after three
iterations.

The contribution of this paper can be summarized as follows: \textcolor{black}{
\begin{itemize}
\item The conventional Rank-1-AFMS of \cite{Bornhorst} is generalized to
the Rank-2-AFMS, where the direct link connections from the source to the destinations are
exploited.
\item In extension of the previous works on rank-one beamforming \cite{Bornhorst}, \cite{12}, \cite{TransmitBeamforming}, \cite{multigroupMulticast}, \cite{Silva}-\cite{ERt}, and on rank-two beamforming \cite{Eusipco}, \cite{Rank2}-\cite{9} for multicasting scenarios, we combine the beamformer design with power allocation.
\item The rank-two beamforming technique of \cite{Rank2} and \cite{Xin} is generalized to a distributed beamforming system with a
fundamental difference in the system model.
\item The proposed Rank-2-AFMS enjoys symbol-by-symbol maximum-likelihood (ML) detection due to the uncorrelated destination noise.
\item A rank-two optimization framework is established to develop the CCCP algorithm. Compared with the traditional SDR approach the latter approach has the following advantages:
\begin{itemize}
\item No additional searches for the optimum source power and the highest minimum SNR are required as the CCCP algorithm computes all parameters jointly.
\item The CCCP algorithm converges to a stationary point.
\end{itemize}
\item The simulation results demonstrate that the performance of the proposed system combined with
the proposed CCCP algorithm is close to the theoretical performance bound.
\end{itemize}}

\newcounter{Satz}

{\it Notation:}
${\rm E}\{\cdot\}$, $|\cdot|$, ${\rm tr}(\cdot)$, $(\cdot)^*$, $(\cdot)^T$, $\Re
\{\cdot\}$, and $(\cdot)^H$ denote the statistical expectation,
absolute value of a complex number, trace of a
matrix, complex conjugate, transpose, real part operator, and Hermitian
transpose, respectively. $\bY \succeq 0$ means that $\bY$ is a positive
semidefinite matrix. ${\rm diag}(\mathbf{a})$ denotes a diagonal matrix, with
the entries of the vector $\mathbf{a}$ on its diagonal, ${\rm
blkdiag}([\mathbf{Y_1,Y_2,\dots,Y_m}])$ is the block diagonal matrix formed from
the matrices $\mathbf{Y_1,Y_2,\dots,Y_m} $. $\mathbf{0}_N$  is the $N\times 1$
vector containing zeros  in all entries.
 $\mathbf{O}_N\triangleq {\rm
diag}(\mathbf{0}_N) $ and $\mathbf{I}_N $ is the $N \times N$ identity matrix.
$\mathbf{x} \sim {\cal
N } (\mathbf{a},\mathbf{Y})$ means that $\mathbf{x}$ is circularly symmetric
complex Gaussian distributed with mean $\mathbf{a}$ and covariance matrix
$\mathbf{Y}$. For the complex number $z=|z|e^{j\phi}$, $ {\rm arg}(z)$ denotes the
phase $\phi $. $\mathcal{R}(\mathbf{Y}) $ denotes the rank of $\mathbf{Y}$.

\section{SYSTEM MODEL} \label{s164}
\begin{figure}
	\psfrag{Transmitter}{source}
	\psfrag{Relays}{relays}
	\psfrag{Receivers}{destinations}
	\psfrag{destinations}{\bf Destinations}
	\psfrag{f}{\small $\bff$}
	\psfrag{g1}{\small $\bg_1$}
	\psfrag{g2}{\small $\bg_M$}
	\psfrag{g_N}{\small $\bg_N$}
\psfrag{d1}{\small $d_1$}
	\psfrag{dM}{\small $d_M$}
\psfrag{d2}{\small $d_2$}
	\psfrag{dM-1}{\small $d_{M-1}$}
\psfrag{f1}{\small $f_1$}
	\psfrag{fR}{\small $f_R$}
  \centerline{\epsfig{figure=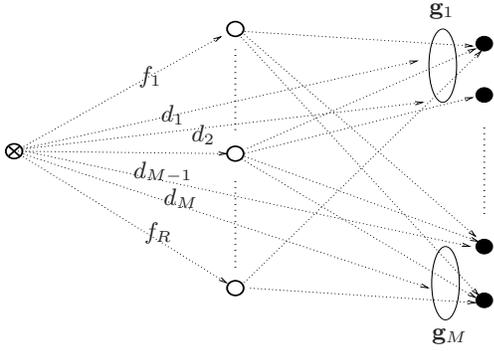,width=6.5cm,height = 4.5cm}}
\caption{Multicasting via the relay network; $\otimes$: source, $\circ $:
relay, $\bullet$: destination.} \label {f185}
\end{figure}
Consider a wireless network of $R$ relays forwarding the signals
from a single source to $M$ destinations. In our single-group multicasting scenario, all destinations demand the same information. The source, the relays, and
the destinations are single antenna devices, see Fig.~\ref{f185}.

\begin{figure}
	\psfrag{1st}{\small First time slot}
	\psfrag{2nd}{\small Second time slot}
	\psfrag{3rd}{\small Third time slot}
	\psfrag{4th}{\small Fourth time slot}
	\psfrag{f}{\small $\bff$}
	\psfrag{g1}{\small $\bg_1$}
	\psfrag{g2}{\small $\bg_M$}
	\psfrag{g_N}{\small $\bg_N$}
  \centerline{\epsfig{figure=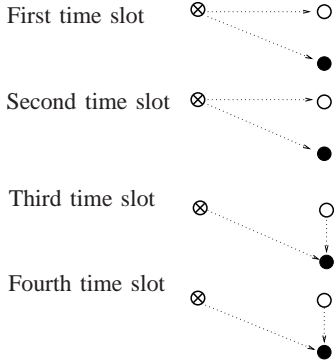,scale = 0.3}}
\caption{Proposed Rank-2-AFMS; $\otimes$: source, $\circ $:
relays, $\bullet$: destinations.} \label {TransmissionScheme}
\end{figure}

In the proposed Rank-2-AFMS, two data symbols are jointly processed in a
four time slot
scheme, see Fig.~\ref{TransmissionScheme}.
In the first and second time slot of the Rank-2-AFMS, the source
transmits the data
symbols $s_1$ and $s_2^\ast$, respectively, which are drawn from a discrete symbol constellation
$\mathcal{S}$. Both symbols are weighted by the same
real-valued power
scaling factor $\alpha_1$. The $R \times 1$ vectors
$\bx_1\triangleq{[x_{1,1},\dots,x_{R,1}]}^T$ and
$\bx_2\triangleq{[x_{1,2},\dots,x_{R,2}]}^T$ of the
received signals at the relays in the first and second time slot are
respectively given by
\begin{align}
	\bx_1 &= \bff \alpha_1 s_{1} + \boldsymbol{\eta}_1,
	 ~~~\bx_2 = \bff \alpha_1 s_{2}^* + \boldsymbol{\eta}_2,
\label{RelayReceived}
\end{align}
where
$\boldsymbol{\eta}_1\triangleq[\eta_{1,1},\dots,\eta_{R,1}]^T$ and
$\boldsymbol{\eta}_2\triangleq[\eta_{1,2},\dots,\eta_{R,2}]^T$ are the $R \times
1$
vectors of the relay noise of the first and second time slot, respectively, and where $
\bff \triangleq {[f_1,\dots,f_R]^T} $
 is the $R \times 1$ vector,
containing the complex coefficients {of the
channels} from the source to the relays.
Note that all channels
in the network are assumed to be frequency flat and constant over the
four considered time slots.

 In
contrast to the system model of \cite{Bornhorst},
we consider that
there exist direct channels from the source to the destinations.
The signals
$y_{m,1}$ and $y_{m,2}$
received by the $m$th destination in the first and second time slot,
respectively, are given by
\begin{equation}
y_{m,1} = d_m \alpha_1 s_1 +\nu_{m,1},~~~~~ y_{m,2} = d_m \alpha_1 s_2^\ast
+\nu_{m,2}, \label{DirectpathFirstTimeslot}
\end{equation}
where $d_m$ is the coefficient of the channel from the source to the $m$th
destination and $\nu_{m,1}$ and $\nu_{m,2}$ represent the destination noise
of the first and second time slot, respectively.

We make the practical assumptions that the noise processes in the network
are
spatially and temporally independent and complex Gaussian distributed.
The noise power at the destinations is given by
\begin{equation}
{\rm E}\{|\nu_{m,q}|^2\}=
\sigma_{\nu}^2,\quad q \in \{1,2,3,4\}, \label{NoiseReceiver}
\end{equation}
and the noise at the relays is distributed according to
\begin{equation}
\boldsymbol{\eta}_1 \sim
{\cal{N}}(\mathbf{0}_R,\sigma_\eta^2 \mathbf{I}_R),~~~ \boldsymbol{\eta}_2
\sim
{\cal{N}} (\mathbf{0}_R,\sigma_\eta^2
\mathbf{I}_R), \label{NoiseDistributionRelay}
\end{equation}
where $\sigma_\eta^2$ is the power of the noise at the relays.

Here, the relays transmit
their
signals over two
different beams. In this fashion, the proposed Rank-2-AFMS enables
the relays to create two different communication links from the source to each
destination. In the conventional Rank-1-AFMS of \cite{Bornhorst}, the
relays transmit their signals over one beam, creating a
single
communication link from the source to each destination.

The $R \times 1$  vectors $\bt_3\triangleq[t_{1,3},\dots,t_{R,3}]^T$ and
$\bt_4\triangleq[t_{1,4},\dots,t_{R,4}]^T$ of the signals transmitted by the
relays in
the third and fourth time slot, respectively, can be expressed as
\begin{align}
	\bt_3 &= \bW_1 \bx_1 + \bW_2 \bx_2^\ast,~~~~
	\bt_4 = -\bW_2 \bx_1^\ast + \bW_1 \bx_2, \label{RelayTrans}
\end{align}
where $\bW_1 \triangleq {\rm diag}(\bw_1^H)$, $\bW_2
\triangleq {\rm diag}(\bw_2^H)$, and
$\mathbf{w}_1=[w_{1,1},\dots,w_{R,1} ]^T $ and
$\mathbf{w}_2=[w_{1,2},\dots,w_{R,2} ]^T $ are the complex $R \times 1$
relay weight vectors. According to equations (\ref{RelayTrans}), the relays transmit their received
signal
vectors
$ \mathbf{x}_{1}$ and $ \mathbf{x}_{2}$ over the two weight vectors
$\mathbf{w}_1 $ and $\mathbf{w}_2$. The use of two weight vectors
increases the degrees of freedom in the distributed
beamformer design. In general, the superposition of multiple symbols that
are
simultaneously transmitted over
different beams leads to difficulties in the decoding at
the destinations. In our proposed Rank-2-AFMS however, symbol-by-symbol detection
at the destinations is facilitated by the  particular spatial encoding scheme applied in equations
(\ref{RelayTrans}). The latter encoding scheme corresponds to the popular Alamouti's OSTBC scheme,
in which two data symbols are transmitted over two spatial channels  over consecutive time slots \cite{Alamouti}. In contrast to conventional OSTBC
schemes, which do not require the availability of  channel state information (CSI) at the transmitter, the Rank-2-AFMS creates two
``artificial'' spatial channels which are shaped by choosing two beams with corresponding beamforming vectors that are designed
based on CSI. It is assumed
that the perfect CSI is available at a central
processing node that computes
the source power and the relay weights and communicates each parameter to its
corresponding node.

Note that according to equations (\ref{RelayTrans}), the
signals transmitted by the relays
consist of linear combinations of the received signals and their conjugates.
This is also the case in
distributed OSTBC schemes which, however, do not use CSI at the transmitter
\cite{DistributedSpaceTime1}, \cite{DistributedSpaceTime2}.

To further exploit the direct links from the source to the destinations, the
source transmits the signals $
\alpha_3 s_1 + \alpha_4 s_2$ and $ - \alpha_4 s_1^\ast +  \alpha_3 s_2^\ast$
in the third and fourth time slot, respectively, where $ \alpha_3$ and
$\alpha_4$
are complex-valued scaling factors.

The
received signals of the third and fourth time slot at the $m$th destination can then
be written as
\begin{align}
	y_{m,3} &= \bg_m^T \bt_3+ d_m \alpha_3 s_1 + d_m\alpha_4 s_2+
\nu_{m,3}, \nonumber \\ \relax
	y_{m,4} &= \bg_m^T \bt_4 - d_m\alpha_4 s_1^\ast +  d_m \alpha_3
s_2^\ast + \nu_{m,4}, \label{ReceivedSignal}
 \end{align}
where
\begin{equation}
\bg_m \triangleq [g_{m,1},\dots,g_{m,R} ]^T \label{gvector}
\end{equation}
is the $R \times 1$  vector
of the complex
frequency flat channels
between the relays and the $m$th destination
and $\nu_{m,3}$ and $\nu_{m,4}$ denote the receiver noise at the $m$th destination
in the third and fourth time slot, respectively.

Introducing $\by_m$ as the vector of the received signals at the $m$th
destination, $\bn_{m}$ as the vector of the
noise at the $m$th destination, and $\bH_{m}$ as the equivalent channel
matrix
and making use of equations (\ref{RelayReceived}), (\ref{RelayTrans}), and
(\ref{gvector}), the received signals of the four time slots in equations
(\ref{DirectpathFirstTimeslot}) and
(\ref{ReceivedSignal}) can be
compactly written as
\begin{align}
\by_m &= \bH_{m} \bs + \bn_m, \label{Systemmodel}
\end{align}
where
\begin{align}
\by_m & \triangleq [y_{m,1}, y_{m,2}^*, y_{m,3}, y_{m,4}^* ]^T,  \\ \relax
\bs & \triangleq [s_{1}, s_{2}]^T, \\
\bn_m & \triangleq  \! \begin{bmatrix} \nu_{m,1} \\
\nu_{m,2}^* \\
 \bw_1^H \bG_m \boldsymbol{\eta}_1 + \bw_2^H \bG_m \boldsymbol{\eta}_2^\ast +
\nu_{m,3} \\ \relax
-\! \bw_2^T \bG_m^H \boldsymbol{\eta}_1\! +\! \bw_1^T \bG_m^H
\boldsymbol{\eta}_2^\ast
+ \nu_{m,4}^*
\end{bmatrix}, \label{noisy}\\ \relax
\bH_{m} &\triangleq\! \begin{bmatrix}
            \alpha_1 d_m & 0 \\
            0 & (\alpha_1 d_m)^\ast \\
           &  \!\!\!\!\!\!\!\!\!\!\!\!\!\!\! \!\!\!\!\! \!\!\!     \bH_{{\rm A},m}
           \end{bmatrix},
 \label{Hmatrix} \\ \relax  \bH_{{\rm A},m} &\triangleq\! \begin{bmatrix}
            h_{m,1} & h_{m,2} \\ \relax
            -h_{m,2}^\ast & h_{m,1}^\ast
           \end{bmatrix} \label{Hmatrix2} \\   \relax
 h_{m,1}  &\triangleq \alpha_1 \bw_1^H \bG_m \bff + \alpha_3 d_m,
\label{channelgains1}
\\ \relax
h_{m,2}   & \triangleq \alpha_1 \bw_2^H \bG_m \bff^\ast + \alpha_4 d_m,
\label{channelgains2} \\ \relax
\bG_m &\triangleq {\rm diag}(\bg_m). \label{GMatrix}
\end{align}
Note that $h_{m,1}$ and $h_{m,2}$ denote the equivalent scalar channel coefficients that model the flat fading channels between the source and the $m$th destination that are obtained from relay beamforming using weight vectors $\bw_1$ and $\bw_2$, respectively. According to equation (\ref{Hmatrix2}) these channels form the matrix $\bH_{{\rm A},m}$, which has the same orthogonal structure as the channel matrix of Alamouti's OSTBC scheme. Interestingly, if no direct source-destination channels exist, i.e., $d_m=0$ for all $m=1,\dots,M$ and the destinations receive signals  only in the third and the fourth time slot, the developed system model of equation (\ref{Systemmodel}) corresponds to the system model of \cite{Rank2} and \cite{Xin} for centralized rank-two beamforming. However, for the latter schemes, the noise at the destinations can be modeled as temporally white.
In contrast to the proposed relay beamforming approach
we observe from equation (\ref{noisy}) that the third and the fourth entry of the noise vector $\mathbf{n}_m$ both depend on the relay noise $\boldsymbol{\eta}_1$ and $\boldsymbol{\eta}_2$ in the first and the second time slot, respectively.
However, due to the specific orthogonal transmission format corresponding to the proposed AFMS it can readily be verified that the covariance matrix of the noise vector in equation (\ref{noisy}) exhibits a diagonal structure of the form
\begin{equation}
{\rm E} \left\{ \bn_m \bn_m^H \right\} = {\rm
blkdiag} \left( [\sigma_{\nu}^2 \mathbf{I}_2,
 \sigma_{m,34}^2 \mathbf{I}_2] \right),
\label{NoiseMatrix}
\end{equation}
where
\begin{align}
\sigma_{m,34}^2 & \triangleq \sigma_\eta^2 (\bw_1^H {\cal G}_m \bw_1 + \bw_2^H
{\cal G}_m
\bw_2)
+\sigma_\nu^2,\label{NoisePower} \\
 { \cal G}_m & \triangleq
\bG_m \bG_m^H  \!=\! {\rm
diag}([|g_{m,1}|^2,\dots,|g_{m,R}|^2]). \label{Def:calG}
\end{align}
The spatially and temporally uncorrelated noise property expressed in equation (\ref{NoiseMatrix}) is essential for the use of simple symbol-by-symbol detection at the destination, which achieves Maximum-Likelihood (ML) performance as will be shown in the following.
To account for the difference in the noise powers of the different time slots due to
noise amplification at the relays, we transform both sides
of equation (\ref{Systemmodel}) by
multiplication with the
diagonal scaling matrix
\begin{equation}
\boldsymbol{\Gamma} \triangleq {\rm
blkdiag} \left([ (
\sigma_{m,34}/\sigma_{\nu}) \mathbf{I}_2, \mathbf{I}_2 ]\right), \label{Umatzrix}
\end{equation}
resulting in the equivalent model representation
\begin{equation}
 \boldsymbol{\Gamma}\by_m = \boldsymbol{\Gamma} \bH_{m} \bs +
\boldsymbol{\Gamma}
\bn_m, \label{UEq}
\end{equation}
corresponding to a uniform noise power model, i.e.,
$\boldsymbol{\Gamma} \bn_m$  is given by
$ {\rm E} \left\{\! \left( \boldsymbol{\Gamma} \bn_m\right) \left(
\boldsymbol{\Gamma} \bn_m
\right)^H \!
\right\} \!\!= \sigma_{m,34}^2 \mathbf{I}_2.$
Then, using equation (\ref{UEq}), the  ML detection problem of finding $\bs$ can be
equivalently reformulated
as the least squares problem
\begin{equation}
 \min_{\bs \in \mathcal{S}\times \mathcal{S}} \lVert \boldsymbol{\Gamma}\by_m -
\boldsymbol{\Gamma} \bH_{m} \bs
\rVert^2. \label{MLprob}
\end{equation}
From equations (\ref{Hmatrix}) and (\ref{Umatzrix}), we observe that
$\mathbf{H}_{m}^H \boldsymbol{\Gamma}^H \boldsymbol{\Gamma} \mathbf{H}_{m}   = c^2 \mathbf{I}_2 $, where
\begin{equation*}
 c  =
\sqrt{(\alpha_1^2 |d_m|^2 \sigma_{m,34}^2)/\sigma_{\nu}^2 +
|h_{m,1}|^2 + |h_{m,2}|^2}. \label{norm}
\end{equation*}
Let us define the matrix $\boldsymbol{\Pi} \triangleq
(1/c)[\boldsymbol{\Gamma} \bH_{m}, \mathbf{Z} ],$ where the matrix $\mathbf{Z}$
is chosen such that $\boldsymbol{\Pi}^H \boldsymbol{\Pi} =
\mathbf{I}_2 $, i.e., equation
(\ref{MLprob}) can be equivalently written as
\begin{align}
 & \min_{\bs \in \mathcal{S} \times \mathcal{S}} \left\lVert
\boldsymbol{\Pi}^H \boldsymbol{\Gamma}\by_m - c
\begin{bmatrix} \mathbf{I}_2 \\
\mathbf{O}_2     \end{bmatrix}
\bs \right\rVert^2 \nonumber \\
= &  \min_{\bs \in \mathcal{S} \times \mathcal{S}} c^2 \lVert
\hat{\mathbf{s}}_m - \bs \rVert^2  + {\rm constant}\label{MLprob2}
\end{align}
where
\begin{equation}
\hat{\mathbf{s}}_m \triangleq \mathbf{B} \by_m = \mathbf{s} + \mathbf{B}
\bn_{m} \label{decoded},
\end{equation}
and
\begin{align}
 &\mathbf{B} \triangleq  (1/c)^2(\boldsymbol{\Gamma}
\bH_{m})^H\boldsymbol{\Gamma}.
\label{BMatrix}
\end{align}
From equation (\ref{MLprob2}), the detection of $s_1$ and $s_2$ decouples into two
scalar detection problems since
\begin{align}
\hat{\mathbf{s}}_m  & \sim
\mathcal{N} \left(\mathbf{s}, \sigma_m^2 \mathbf{I}_2
\right), \label{NoiseDistibutionReceiver} \\
\sigma_m^2 & \triangleq \frac{\sigma_{m,34}^2 \sigma_{\nu}^2}{\alpha_1^2
|d_m|^2\sigma_{m,34}^2 +\sigma_{\nu}^2(|h_{m,1}|^2+|h_{m,2}|^2) },
\end{align}
which follows directly from  equations (\ref{NoiseMatrix}) and (\ref{BMatrix}). In other
words, in the proposed four phase scheme
 ML detection  reduces to simple symbol-by-symbol detection. We remark that
the diagonal structure of the error covariance matrix follows from the
orthogonal encoding in equations (\ref{RelayTrans}).

From equation (\ref{NoiseDistibutionReceiver}), we see that the SNR in
$\hat{\mathbf{s}}_m$ is given by
\begin{equation}
 {\rm SNR}_m = \frac{1}{\sigma_m^2} =
\frac{\alpha_1^2 |d_m|^2}{\sigma_{\nu}^2} +
\frac{|h_{m,1}|^2+|h_{m,2}|^2}{\sigma_{m,34}^2},\label{SNR}
\end{equation}
for both
data symbols. For the sake of
convenience,
let us introduce the following vector notation
\begin{align}
\mathbf{w} & \triangleq [\tilde{\mathbf{w}}_1^T,\tilde{\mathbf{w}}_2^T]^T,
\label{startNot}\\ \relax
\tilde{\mathbf{w}}_1 & \triangleq [
{\mathbf{w}}_1,(\alpha_3  \alpha_1)^\ast]^T,\quad
\tilde{\mathbf{w}}_2
\triangleq [ \mathbf{w}_2^T, (\alpha_4  \alpha_1)^\ast
]^T, \\ \relax
\mathbf{R}_m  &\triangleq {\rm
blkdiag}([\tilde{\mathbf{R}}_m,\tilde{\mathbf{R}}_m]), \label{Rm} \\ \relax
\tilde{\mathbf{R}}_m &  \triangleq {\rm
blkdiag}([\sigma_\eta^2{\cal G}_m,0]). \label{endNot}
\end{align}
Using equations  (\ref{NoisePower}) and
(\ref{startNot}) - (\ref{endNot}) yields
\begin{align}
 \sigma_{m,34}^2 & = \tilde{\mathbf{w}}_1^H \tilde{\mathbf{R}}_m
\tilde{\mathbf{w}}_1+\tilde{\mathbf{w}}_2^H \tilde{\mathbf{R}}_m
\tilde{\mathbf{w}}_2 + \sigma_\nu^2 \label{noseVEc}
\\ \relax &=\mathbf{w}^H \mathbf{R}_m \mathbf{w} + \sigma_\nu^2.
\label{noiesVec}
\end{align}
Let us furthermore introduce
\begin{align}
\mathbf{Q}_m  &\triangleq {\rm
blkdiag}([\tilde{\mathbf{Q}}_{m,1},\tilde{\mathbf{Q}}_{m,2}]), \label{Qm}
\\ \relax
\tilde{\mathbf{Q}}_{m,1} &\triangleq \mathbf{q}_{m,1}\mathbf{q}_{m,1}^H, ~~
\tilde{\mathbf{Q}}_{m,2} \triangleq
\mathbf{q}_{m,2}\mathbf{q}_{m,2}^H,\label{MatrixDefinition}\\ \relax
\mathbf{q}_{m,1} & \triangleq [(\bG_m \bff)^T,d_m]^T, \quad \mathbf{q}_{m,2}
 [(\bG_m \bff^\ast)^T,d_m]^T, \label{channelvectors} \relax \\
 a &\triangleq 1 / |\alpha_1|^2, \label{a_def}
\end{align}
where $a$ is a power scaling factor. With equations (\ref{channelgains1}),
(\ref{channelgains2}), and  (\ref{Qm}) - (\ref{a_def}), we have
\begin{align*}
|h_{m,1}|^2+|h_{m,2}|^2 & = \left( \tilde{\mathbf{w}}_1^H
\tilde{\mathbf{Q}}_{m,1}
\tilde{\mathbf{w}_1} + \tilde{\mathbf{w}}_2^H \tilde{\mathbf{Q}}_{m,2}
\tilde{\mathbf{w}}_2 \right)/a\\ \relax &= \mathbf{w}^H \mathbf{Q}_m
\mathbf{w}/a.
\end{align*}
Using the above identity together with equations (\ref{noseVEc}) and (\ref{noiesVec}),
we reformulate the SNR given in equation (\ref{SNR}) in vector notation as
\begin{align}
& {\rm SNR}_m (\tilde{\mathbf{w}}_1,\tilde{\mathbf{w}}_2,a) \nonumber \\ \relax&
=
\frac{\tilde{\mathbf{w}}_1^H \tilde{\mathbf{Q}}_{m,1}
\tilde{\mathbf{w}}_1+\tilde{\mathbf{w}}_2^H
\tilde{\mathbf{Q}}_{m,2}
\tilde{\mathbf{w}}_2}{(\tilde{\mathbf{w}}_1^H
\tilde{\mathbf{R}}_m \tilde{\mathbf{w}}_1+\tilde{\mathbf{w}}_2^H
\tilde{\mathbf{R}}_m \tilde{\mathbf{w}}_2 + \sigma_\nu^2)a} +
\frac{|d_m|^2}{\sigma_\nu^2 a}, \label{SNR_Reformulated_new}
\end{align}
or, equivalently, as
\begin{equation}
{\rm SNR}_m (\mathbf{w},a) =\! \frac{\mathbf{w}^H \mathbf{Q}_m
\mathbf{w}}{(\mathbf{w}^H
\mathbf{R}_m \mathbf{w} + \sigma_\nu^2)a} +
\frac{|d_m|^2}{\sigma_\nu^2 a}. \label{SNR_Reformulated}
\end{equation}

\section{Beamformer Design and Power Control} \label{s218}
In this section, we derive an algorithm to design the weight vectors and to
distribute the power between different time slots and between the source and the
relays. We consider the problem of maximizing the minimum QoS measured in terms
of the SNR at
the destinations subject to power constraints.
The power constraints include thresholds on the individual power of each relay and the source, the sum power of the relays, and the total power of the network.
Maximizing the minimum SNR is a practical objective, e.g., for packet-data traffic with full buffer networks the receivers demand the largest feasible data rates rather than a certain constant data rate, provided that the average rate is satisfactory \cite{4G}.
Hence, max-min fairness is a common design criterion that has been used in \cite{Chen08}, \cite{MultiUserTwoWay}, \cite{Schad2}, \cite{multigroupMulticast}, \cite{1}, and \cite{Dartman2}.
Here, the corresponding max-min fairness optimization problem is formulated as
\vspace{2mm}
\begin{center}
\fbox{\parbox[c]{8.25cm}{
{$\mathcal{P}$:}
\begin{eqnarray*}
\max_{(\mathbf{w},a)}  \min_{m\in \{1,\dots,M\}}  &{\rm SNR}_m  \nonumber
\\ \relax {\rm
s.t.} \quad &  (\mathbf{w},a) \in \Omega, \label{Opt:SNR_Maximization}
\end{eqnarray*}
}}
\end{center}
\vspace{2mm}
where $\mathbf{w}$ and $a$  belong to the set \begin{equation*}
 \Omega \triangleq \bigl\{ \mathbf{w},  a  \mid \mathbf{w},a
\textnormal{  satisfy (\ref{POwerConstraintsOriginal})}
\bigr\}
\end{equation*} characterized by the following
power constraints
\begin{subequations}
\label{POwerConstraintsOriginal}
\begin{align}
\textnormal{positivity: } &   a > 0,  \label{pospower} \\ \relax
\textnormal{individual relay power: } &   p_{r}(\mathbf{w},a) \leq p_{r,\max}\quad \forall r \in \{1,\dots,R\},
 \label{IndPowerConstraint} \\ \relax
\textnormal{relay sum power: } & \sum_{r=1}^R p_r  (\mathbf{w},a)\leq P_{R,\max},
\label{SumPowerConstraint}\\ \relax
\textnormal{source power: } & P_S(\mathbf{w},a) \leq P_{S,\max}, \label{SourcePowerConstraint}\\ \relax
\textnormal{total power: }& P_T(\mathbf{w},a) \leq P_{T,\max}. \label{overallPowerConstraint}
\end{align}
\end{subequations}
The positivity condition  (\ref{pospower}) results from the parameter transformation  in definition (\ref{a_def}) and ensures that the source power is a real and
positive number.
In inequality (\ref{IndPowerConstraint}),
$p_r(\mathbf{w},a)$ represents the transmit power of the $r$th relay,
$p_{r,\max} $ is the maximum individual power
value for the $r$th relay,
$ \sum_{r=1}^R p_r (\mathbf{w},a)$ in inequality (\ref{SumPowerConstraint}) is the sum
power
transmitted by the relays in one time slot for which the power threshold value $P_{R,\max} $ applies,
$P_S(\mathbf{w},a)$  in inequality (\ref{SourcePowerConstraint}) is the total transmitted
power
at the source in four consecutive time slots with threshold value $P_{S,\max}$,
$P_{T}(\mathbf{w},a) $ in inequality (\ref{overallPowerConstraint}) is the total
transmit power of
the
network, including source and relay powers,
during four time slots,
and $P_{T,\max}$ is the threshold value for
$P_{T}(\mathbf{w},a)$.

Using equations (\ref{RelayReceived}) and
(\ref{RelayTrans}), the power $p_r(\mathbf{w},a)$
transmitted by the $r$th relay in the third time slot can
be
derived as
\begin{align}
&p_{r}(\mathbf{w},a)  = E\{| t_{r,3}|^2\} \nonumber \\ &= E\{|\alpha_1
w_{r,1}^\ast f_r
s_1+w_{r,1}^\ast \eta_{r,1} +
\alpha_1^\ast
w_{r,2}^\ast f_r^\ast s_2 +w_{r,2}^\ast \eta_{r,2}^\ast|^2\}
\nonumber \\ \relax & = \left(|w_{r,1}|^2 +
|w_{r,2}|^2 \right) ( |\alpha_1  f_r|^2+\sigma_\eta^2) \nonumber \\
& = \tilde{\mathbf{w}}_1^H \tilde{{\bD}}_r \tilde{\mathbf{w}}_1/a +
\tilde{\mathbf{w}}_2^H \tilde{{\bD}}_r \tilde{\mathbf{w}}_2/a +
\tilde{\mathbf{w}}_1^H \tilde{{\mathbf{E}}}_r \tilde{\mathbf{w}}_1 +
\tilde{\mathbf{w}}_2^H \tilde{{\mathbf{E}}}_r \tilde{\mathbf{w}}_2\nonumber \\
&=
\mathbf{w}^H {\bD}_r \mathbf{w}/a + \mathbf{w}^H {\mathbf{E}}_r
\mathbf{w},\label{DefRelay}
\end{align}
where the $(R+1) \times
(R+1) $ matrices
$\tilde{\bD}_r$ and $\tilde{\mathbf{E}}_r$ have respectively $ |f_r|^2 $
and $\sigma_\eta^2 $ as their $r$th
 diagonal entry and zeros elsewhere, and where ${\bD}_r
\triangleq {\rm
blkdiag}([\tilde{\bD}_r,\tilde{\bD}_r])$ and ${\bE}_r \triangleq {\rm
blkdiag}([\tilde{\bE}_r,\tilde{\bE}_r])$. In equations (\ref{DefRelay}),
we have used the assumption that the data symbols are independent and identically distributed
with zero mean and unit variance. Due to the symmetry in the transmission scheme, the relay power in
the fourth time slot is equivalent to the relay power of the third time slot,
i.e., $p_r(w,a) = E\{ |t_{r,3}|^2\}= E\{ |t_{r,4}|^2\}.$ Hence $p_r(w,a)$ represents
the
relay power consumed in each time slot in which the relays transmit. Note that the inequality constraints (\ref{IndPowerConstraint})
are convex as $p_{r}(\mathbf{w},a)$ in equations (\ref{DefRelay}) is expressed as the
sum of the convex quadratic form
$\mathbf{w}^H
{\mathbf{E}}_r
\mathbf{w}$ as well as the fraction of the convex
quadratic form $\mathbf{w}^H {\bD}_r \mathbf{w}$ and the linear term $a$, which
is a
convex function \cite{Boyd}. The same holds true for the condition
$\sum_{r=1}^R
p_{r}(\mathbf{w},a) \leq P_{R,\max}$ in inequality
(\ref{SumPowerConstraint}), as the summation of
convex functions yields a convex
function.  The  transmit power of the
source during the four
time slots is given by
\begin{align}
P_S (\mathbf{w},a) &= E\{| \alpha_1 s_1 |^2\} + E\{| \alpha_1 s_2^\ast |^2\}
\!+\!
E\{| \alpha_3
s_1 + \alpha_4 s_2 |^2\} \nonumber \\ \relax&+  E\{| - \alpha_4 s_1^\ast +
\alpha_3
s_2^\ast|^2\} ={2}/{a} + 2 |\alpha_3 |^2+ 2 |\alpha_4 |^2 \nonumber \\ \relax
& =  2/a +\tilde{\mathbf{w}}_1^H \tilde{{\mathbf{S}}}
\tilde{\mathbf{w}}_1/a+\tilde{\mathbf{w}}_2^H \tilde{{\mathbf{S}}}
\tilde{\mathbf{w}}_2/a\nonumber \\
& = 2/a + \mathbf{w}^H {\mathbf{S}} \mathbf{w}/a,
\end{align}
where $\tilde{\mathbf{S}}$ is an $(R+1) \times
(R+1) $  matrix, having $2$ as its $(R+1)$th diagonal entry
and zeros elsewhere and where $ \mathbf{S} \triangleq {\rm
blkdiag}([\tilde{\mathbf{S}},\tilde{\mathbf{S}}])$. Note that
$P_S(\mathbf{w},a)$
is a
convex function of $\mathbf{w}$ and
$a$.

As the sum powers of the
relays of the third and fourth time slot are equal, the
total transmit power of the relays and the source during four time slots amounts to
\begin{equation}
 P_{T}(\mathbf{w},a) = P_S(\mathbf{w},a) + 2 \sum_{r=1}^R p_r (\mathbf{w},a).
\label{DefPT}
\end{equation}

With equations (\ref{DefRelay}) - (\ref{DefPT}), the
power constraints (\ref{pospower}) - (\ref{overallPowerConstraint}) can be reformulated as
\begin{subequations}
\label{POwerConstraintsVec}
\begin{align}
& a > 0,   \label{constraint_a}\\ \relax
& \mathbf{w}^H {\bD}_r \mathbf{w}/a + \mathbf{w}^H {\mathbf{E}}_r
\mathbf{w} \leq p_{r,\max}\quad \forall r \in \{1,\dots,R\}, \label{constraint:Ind_Power}
 \\ \relax
& \sum_{r=1}^R \left(\mathbf{w}^H {\bD}_r \mathbf{w}/a + \mathbf{w}^H
{\mathbf{E}}_r
\mathbf{w} \right) \leq P_{R,\max}, \\ \relax
& 2/a + \mathbf{w}^H {\mathbf{S}} \mathbf{w}/a \leq P_{S,\max}, \label{constraint:Source_Power} \\
\relax
& 2/a + \mathbf{w}^H {\mathbf{S}} \mathbf{w}/a \!+\!2\sum_{r=1}^R  \left( \!
\mathbf{w}^H
{\bD}_r \mathbf{w}/a \!+\ \mathbf{w}^H {\mathbf{E}}_r
\mathbf{w} \!\right) \leq
P_{T,\max},
\label{Opt:Powerconstraints2}
\end{align}
\end{subequations}
respectively.
Note that $\mathcal{P}$  represents the optimization problem in
its general form. We remark that the optimization
procedures  developed in the following for problem $\mathcal{P}$ remain valid if some constraints
(\ref{IndPowerConstraint}) - (\ref{overallPowerConstraint}) are removed.

Introducing the auxiliary variable $t$, problem $\mathcal{P}$  can
equivalently be written as
\vspace{2mm}
\begin{center}
 \fbox{\parbox[c]{8.25cm}{
$\mathcal{M}:$
\begin{align*}
\min_{\mathbf{w},a,t}  \quad & t \nonumber \\ \relax
 {\rm s.t.} \quad & t > 0,  \nonumber \\ \relax
& {\rm SNR}_m(\mathbf{w},a) \geq 1/t \quad \forall m \in \{1,\dots,M\},
\nonumber \\ \relax
& (\mathbf{w},a) \in \Omega,
\end{align*}
}}
\end{center}
\vspace{2mm}
where $1/t$ represents the minimum SNR at the destinations.

The difficulty associated with solving problem $\mathcal{M}$ lies in the SNR
constraints which can be formulated as
\begin{align}
& \lambda_m ( \mathbf{w}, a,  t) \triangleq \nonumber \\ &\frac{\mathbf{w}^H
\mathbf{R}_m \mathbf{w} \!+\! \sigma_\nu^2}{t} -
\frac{\mathbf{w}^H (\mathbf{Q}_m
\!+\!(|d_m|^2 /\sigma_\nu^2)\mathbf{R}_m)
 \mathbf{w} \!+\! |d_m|^2 }{a} \!\leq\! 0,\label{SNR-Constraint}
\end{align}
where we have used
the SNR expression of equation (\ref{SNR_Reformulated}). Due to the negative term on the
left hand side of the above inequality, the SNR constraints are non-convex in
general.

\subsection{Rank-two property and relation of the Rank-2-AFMS to the
Rank-1-AFMS} \label{sR2}
The Rank-1-AFMS can be regarded as a special
case of
the Rank-2-AFMS where symbols are transmitted sequentially by a single
beamformer, i.e., choosing $\mathbf{w}_2 = \mathbf{0}_{R+1}$. In the Rank-1-AFMS, each
symbol
is
communicated in two time slots; in the first time slot the source sends
the
signal to the relays and in the second time slot the relays forward their
received signals to the  destinations.
Note that the number of time slots used to transmit
one data symbol is two for the Rank-1-AFMS ($\bw_2 =
\mathbf{0}_{R+1}$)
and the proposed Rank-2-AFMS ($\bw_2 \neq \mathbf{0}_{R+1}$).

In the following we analyze the Rank-1-AFMS and the Rank-2-AFMS,
applying
SDR to problem $\mathcal{M}$. We will prove that for both schemes, the
SDR versions of $\mathcal{M}$ are equivalent.

Let us derive an equivalent representation of the SNR at the $m$th
destination
given by equation  (\ref{SNR_Reformulated_new}).
Defining the matrix $\mathbf{A} \triangleq {\rm
diag}([e^{2j\varphi_{1},},\dots,e^{2j\varphi_{R}},1])$, where $\varphi_{r}
\triangleq{\rm
arg}(f_r)$, we notice from the definitions (\ref{Qm}) and
(\ref{MatrixDefinition}) that
$\tilde{\mathbf{Q}}_{m,1} =\mathbf{A}^H \tilde{\mathbf{Q}}_{m,2} \mathbf{A}$.
The unitary
transformation $\hat{\mathbf{w}}_2 \triangleq
\mathbf{A}\tilde{\mathbf{w}}_2$  exhibits the useful
property that
$ \tilde{\mathbf{w}}_2^H
\tilde{\mathbf{Q}}_{m,2}
\tilde{\mathbf{w}}_2  = \hat{\mathbf{w}}_2^H
\tilde{\mathbf{Q}}_{m,1}
\hat{\mathbf{w}}_2$. Moreover, $
\tilde{\mathbf{w}}_2^H
\tilde{\mathbf{R}}_{m}
\tilde{\mathbf{w}}_2 =\hat{\mathbf{w}}_2^H
\tilde{\mathbf{R}}_{m}
\hat{\mathbf{w}}_2, $ which follows from the definition of $\mathbf{A}$ and
from equations (\ref{Def:calG}) and (\ref{endNot}). Then, using equation
(\ref{SNR_Reformulated_new}), the SNR constraints  (\ref{SNR-Constraint})
 can be formulated as
\begin{align} & \tilde{\mathbf{w}}_1^H \tilde{\mathbf{Q}}_{m,1}
\tilde{\mathbf{w}}_1+\hat{\mathbf{w}}_2^H
\tilde{\mathbf{Q}}_{m,1}
\hat{\mathbf{w}}_2 \geq \nonumber \\ \relax &\left(\frac{a}{t} -
\frac{|d_m|^2}{\sigma_\nu^2 }\right)(\tilde{\mathbf{w}}_1^H
\tilde{\mathbf{R}}_m \tilde{\mathbf{w}}_1+\hat{\mathbf{w}}_2^H
\tilde{\mathbf{R}}_m \hat{\mathbf{w}}_2 + \sigma_\nu^2).
\nonumber
 \end{align}
Due to the quadratic forms $\tilde{\mathbf{w}}_1^H
\tilde{\mathbf{Q}}_{m,1}
\tilde{\mathbf{w}}_1$ and $\hat{\mathbf{w}}_2^H
\tilde{\mathbf{Q}}_{m,1}
\hat{\mathbf{w}}_2$, the above inequality describes a non-convex set.
Non-convex problems are difficult to solve and generally
NP-hard. To reformulate non-convex QCQPs, the SDR technique has been
proposed in the literature
\cite{TransmitBeamforming}, \cite{multigroupMulticast}.  In the latter technique,
the identity $\tilde{\mathbf{w}}_1^H
\tilde{\mathbf{Q}}_{m,1}
\tilde{\mathbf{w}}_1 = {\rm
tr}(\tilde{\mathbf{w}}_1 \tilde{\mathbf{w}}_1^H \tilde{\mathbf{Q}}_{m,1}) $ is
exploited and $ \tilde{\mathbf{w}}_1 \tilde{\mathbf{w}}_1^H $ is substituted in $\cal M$  by a
positive semidefinite matrix $\mathbf{X}_1$. Substituting further $\hat{\mathbf{w}}_2\hat{\mathbf{w}}_2^H$ by
${\mathbf{X}}_2 $ leads to
\begin{align}
\min_{{\mathbf{X}}_1,{\mathbf{X}}_2,t,a}  \quad & t \nonumber \\ \relax
 {\rm s.t.} \quad & t > 0,  \nonumber \\ \relax
&   {\rm tr} \bigl( \left(\mathbf{X}_1  +\mathbf{X}_2 \right)
\tilde{\mathbf{Q}}_{m,1} \bigr)\geq \nonumber \\ \relax &\left(\frac{a}{t} -
\frac{|d_m|^2}{\sigma_\nu^2 }\right)\left({\rm tr} ( (\mathbf{X}_1
+\mathbf{X}_2)
\tilde{\mathbf{R}}_{m}) + \sigma_\nu^2\right) \nonumber \\& \forall m \in
\{1,\dots,M\},
\nonumber \\ \relax
& (\mathbf{X}_1 + \mathbf{X}_2, a) \in \Upsilon, \nonumber  \\
& \mathbf{X}_1 \succeq 0, {\mathbf{X}}_2 \succeq 0,  \label{Opt:MaxMinSDR}
\end{align}
where the constraints ${\rm
rank}(\mathbf{X}_1 )= 1$ and ${\rm
rank}(\mathbf{X}_2 )= 1$ are neglected and where the set $$\Upsilon \triangleq
\bigl\{
\mathbf{X}, a  \mid \mathbf{X},a
\textnormal{  satisfy
(\ref{SDRPOwerConstraints})}
\bigr\}$$ is defined by the power constraints
\begin{subequations}
\label{SDRPOwerConstraints}
\begin{align}
& a > 0, \label{constraint_a_SDR}\\ \relax
& {\rm tr} \bigl(\mathbf{X}(
\tilde{\mathbf{D}}_{r}/a +\tilde{\mathbf{E}}_{r} ) \bigr) \leq
p_{r,\max}\quad \forall r \in \{1,\dots,R\}, \label{constraint_pr_SDR}
\\ \relax
& \sum_{r=1}^R {\rm tr} \bigl(\mathbf{X}
(
\tilde{\mathbf{D}}_{r}/a +\tilde{\mathbf{E}}_{r} ) \bigr) \leq P_{R,\max},
\label{constraint_PR_SDR} \\ \relax
& 2/a + {\rm tr} \bigl( \mathbf{X}
\tilde{\mathbf{S}} \bigr)/a \leq P_{S,\max},\\
\relax
& 2/a + {\rm tr} \bigl( \mathbf{X}
\tilde{\mathbf{S}} \bigr)/a +2\sum_{r=1}^R {\rm tr} \bigl(\mathbf{X}
(
\tilde{\mathbf{D}}_{r}/a +\tilde{\mathbf{E}}_{r})\bigr) \leq
P_{T,\max}.
\label{Opt:PowerconstraintsSDR}
\end{align}
\end{subequations}
Here the constraints (\ref{constraint_a_SDR}) - (\ref{Opt:PowerconstraintsSDR}) correspond
to constraints
(\ref{constraint_a}) - (\ref{Opt:Powerconstraints2}), respectively.
 Note that for the Rank-1-AFMS
$\hat{\mathbf{w}}_2 = \mathbf{0}_{R+1}$ follows from ${\mathbf{w}}_2 =
\mathbf{0}_{R+1}$ and
 $\mathbf{X}_2 = \mathbf{O}_{R+1}$ holds true. For the Rank-1-AFMS, the SDR
problem corresponding to problem  (\ref{Opt:MaxMinSDR})  is given by
\begin{align}
\min_{{\mathbf{X}}_1,t,a}  \quad & t \nonumber \\ \relax
 {\rm s.t.} \quad & t > 0,  \nonumber \\ \relax
&   {\rm tr} \bigl( \mathbf{X}_1
\tilde{\mathbf{Q}}_{m,1} \bigr)\geq \left(\frac{a}{t}
-
\frac{|d_m|^2}{\sigma_\nu^2 }\right)\!\!\left({\rm tr} \bigl( \mathbf{X}_1
\tilde{\mathbf{R}}_{m} \bigr) \!+\! \sigma_\nu^2\right) \nonumber\\& \forall m \in
\{1,\dots,M\},
\nonumber \\ \relax
& (\mathbf{X}_1,a) \in \Upsilon, \nonumber  \\
& \mathbf{X}_1 \succeq 0.  \label{Opt:MaxMinSDRAFMS}
\end{align}
\newtheorem{Bodo}[Satz]{Theorem}
\begin{Bodo}
\textbf{(Equivalence of SDR problems)} The optimum value $t^\star$ of the problem (\ref{Opt:MaxMinSDR}) corresponding
to the Rank-2-AFMS is the same as for
problem (\ref{Opt:MaxMinSDRAFMS}) corresponding to the Rank-1-AFMS.
\end{Bodo}
\begin{proof}
Let $(\mathbf{X}_1^\star,t^\star,a^\star)$ be a solution to problem
(\ref{Opt:MaxMinSDRAFMS}), then
$(\mathbf{X}_1^\star,\mathbf{X}_2=\mathbf{O}_{R+1},t^\star,a^\star)$ is a
solution to problem
(\ref{Opt:MaxMinSDR}) as the
constraint functions in problem (\ref{Opt:MaxMinSDR}) only depend on the sum of
$\mathbf{X}_1$ and $\mathbf{X}_2$.
On the other hand, if $(\mathbf{X}_1^\star,\mathbf{X}_2^\star,t^\star,a^\star)$
is a solution
of problem (\ref{Opt:MaxMinSDR}), then
$(\mathbf{X}_1^\star+\mathbf{X}_2^\star,t^\star,a^\star)$ is a
solution of problem (\ref{Opt:MaxMinSDRAFMS}) as the sum of positive semidefinite
matrices results in a positive semidefinite matrix.
\end{proof}
As a consequence of the latter theorem, the SDR of $\cal{M}$
for the
Rank-1-AFMS given by problem (\ref{Opt:MaxMinSDRAFMS}) and the SDR of ${\cal M}$ for the
Rank-2-AFMS given by problem (\ref{Opt:MaxMinSDR}) are equivalent. For the Rank-2-AFMS,
a feasible solution for
$\mathcal{M}$ can be computed from the eigendecomposition of $\bX_1^\star$ if $\mathcal{R}(\mathbf{X_1}^\star)\leq 2$  as will be shown in the following.

Let $\mathbf{X_1}^\star $  be a solution of problem (\ref{Opt:MaxMinSDRAFMS}) with $\mathcal{R}(\mathbf{X_1}^\star)= 2$ and let
$\mathbf{X}_1^\star = \lambda_1
\mathbf{u}_1 \mathbf{u}_1^H + \lambda_2
\mathbf{u}_2 \mathbf{u}_2^H $ with non-zero eigenvalues  $\lambda_1$ and
$\lambda_2$ and the respective eigenvectors $\mathbf{u}_1 $ and $\mathbf{u}_2
$. Then, the solution $\mathbf{X}_1^\star $ of the
relaxed problem is a feasible global solution of $\mathcal{P}$. The
decomposition of
$\mathbf{X}_1^\star$ into two components is not unique and can be obtained from
any
vector pair $\tilde{\mathbf{w}}_1$ and $ \hat{\mathbf{w}}_2$ satisfying
$\mathbf{X}_1^\star =
\tilde{\mathbf{w}}_1 \tilde{\mathbf{w}}_1^H + \hat{\mathbf{w}}_2
\hat{\mathbf{w}}_2^H $. This is different from the Rank-1-AFMS approach
where a single beamforming vector
$\tilde{\mathbf{w}}_1$ is computed and where the SDR solution is only feasible if $\mathcal{R}(\mathbf{X_1}^\star) = 1$.  In the proposed
Rank-2-AFMS, the number of degrees of freedom is increased due to the
introduction of two linearly independent weight vectors.

We remark that problem (\ref{Opt:MaxMinSDRAFMS}) still represents a non-convex
problem, as the fractions of two linear terms in inequalities (\ref{constraint_pr_SDR}),
(\ref{constraint_PR_SDR}), and (\ref{Opt:PowerconstraintsSDR}) represent non-convex
functions. Moreover, the multiplication of the two linear terms in the
SNR constraints (\ref{Opt:MaxMinSDRAFMS}) results in a non-convex function.

\textcolor{black}{However, for constant $t$ and $a$, the { \it feasibility problem}
to compute a  feasible matrix ${\mathbf{X}}_1$  is a convex SDP. In \cite{multigroupMulticast}, the authors have proposed to perform a 1D bisection search to find
the optimum $t$ for a problem similar to problem (\ref{Opt:MaxMinSDRAFMS}) in the context of transmit
beamforming. In the latter search procedure, the feasibility problem is solved in each iteration.
Here, in order to solve problem (\ref{Opt:MaxMinSDRAFMS}), a 2D search
is required
since there is an additional power scaling factor $a$ involved. Note that the
optimum power
allocation factor $a^\star$ cannot be found by bisection search as it is
a-priori not possible to determine an appropriate search interval.}
\textcolor{black}{Therefore, we propose to perform grid search on $a$. For each point on the search grid, an additional bisection search according to \cite{multigroupMulticast} on the optimum $t$ is performed.
We refer to this search procedure as the SDR2D algorithm.}

One drawback of the SDR2D algorithm is its computational burden to solve an
SDP in each iteration of the exhaustive 2D search. Another drawback is that
$\mathcal{R}(\mathbf{X}_{1}^\star)>2$ in general. In this case, the
corresponding optimal value $t^\star$ is only a lower bound for $\mathcal{M}$ (within the
grid search precision) as
$\mathbf{X}_{1}^\star$ is not feasible for problem $\mathcal{P}$ since there
exists
no rank-two
decomposition $\mathbf{X}_1^\star =
\tilde{\mathbf{w}}_1 \tilde{\mathbf{w}}_1^H + \hat{\mathbf{w}}_2
\hat{\mathbf{w}}_2^H $.  For cases in which the SDR solution is not feasible for
the
original problem,  randomization techniques have been applied in
\cite{TransmitBeamforming}--\cite{10} to generate
feasible points that are suboptimal, in general.

Besides the SDR-based algorithms to compute vectors, recently,
iterative algorithms have been developed which outperform the SDR-based
randomization algorithms in terms of performance and computational complexity
\cite{MultiUserTwoWay}--\cite{Eusipco}, \cite{ERt}. The latter
algorithms perform an inner approximation of the original non-convex problem
and maintain the number of variables,
whereas in the SDR technique, the number of variables is roughly squared.

In the rest of this section, we develop an iterative algorithm which  computes
the weight vector and adjusts the source power
to maximize the minimum received SNR at the destinations.

\subsection{Convex inner approximation technique} \label{sec:SNR_Maximization}
In the max-min fairness beamforming problem $\mathcal{M}$,
the left hand side of the constraints (\ref{SNR-Constraint}) is the
difference of $({\mathbf{w}^H
\mathbf{R}_m \mathbf{w} + \sigma_\nu^2})/{t}$ and $
({\mathbf{w}^H (\mathbf{Q}_m
+(|d_m|^2 /\sigma_\nu^2) \mathbf{R}_m)
 \mathbf{w} + |d_m|^2 })/{a}$. As $\mathbf{R}_m$ and $\mathbf{Q}_m$ are
positive-semidefinite matrices, these functions are both convex since they
consist of a convex  quadratic form divided by a linear term
and a constant divided by a linear term \cite{Boyd}. Therefore, the left hand
side
of inequality
(\ref{SNR-Constraint}) is a difference of two convex (DC) functions
and $\mathcal{M}$ belongs to the class of DC programs
\cite{Horst}--\cite{CCCP4}.

To solve $\mathcal{M}$ approximately, we propose an iterative algorithm
which generates a sequence of weight vectors $\mathbf{w}^{(k)}$, with iteration index $k \in \mathbb{N}_0$. At iteration $k+1$, the subsequent vector $\mathbf{w}^{(k+1)}$ of
$\mathbf{w}^{(k)}$ is generated according to
\begin{equation}
\mathbf{w}^{(k+1)} \triangleq\mathbf{w}^{(k)} + \Delta \mathbf{w}  \label{map1}
\end{equation}
using the update vector $\Delta \mathbf{w}$ as further specified below. Similarly, $a^{(k)}$ and $t^{(k)}$
are the
optimization variables $a$ and $t $ at the $k$th iteration, respectively. They
are  updated according to
\begin{equation}
a^{(k+1)} \triangleq  a^{(k)}+\Delta a, ~~t^{(k+1)} \triangleq
t^{(k)} +\Delta t, \label{map2}
\end{equation}
where $\Delta a$ and $\Delta t$  are the update variables as given below.

Let us assume  that $ \mathbf{w}^{(k)} $, $  a^{(k)}$, and $t^{(k)} $ represent fixed feasible points of problem $\mathcal{M}$
and let $\Delta \mathbf{w}$, $\Delta a$, and $ \Delta t$ denote optimization
variables.
If we replace $ \mathbf{w}$ in $\mathcal{M}$ by $
\mathbf{w}^{(k)} +\Delta \mathbf{w}$, $a$ by $a^{(k)}+\Delta a$, and $t$ by $t^{(k)}+\Delta t$, the power constraints defined by the set $\Omega $ remain convex. To
derive a convex approximation of the constraints (\ref{SNR-Constraint})
let us replace the concave part by its first order Taylor approximation around
$(\mathbf{w}^{(k)},a^{(k)},t^{(k)})$, resulting in the convex constraint
\begin{align}
& \bar{\lambda}_m^{(k)}(\Delta \mathbf{w},\Delta a, \Delta t) \triangleq
\frac{( \mathbf{w}^{(k)} +\Delta
\mathbf{w})^H
\mathbf{R}_m ( \mathbf{w}^{(k)}+\Delta \mathbf{w}) + \sigma_\nu^2}{t^{(k)} +
\Delta t}  \nonumber \\
&+\frac{\mathbf{w}^{(k)H}
(\mathbf{Q}_m
\!+\!(|d_m|^2 /\sigma_\nu^2)\mathbf{R}_m) \mathbf{w}^{(k)}\! + \!|d_m|^2
}{a^{(k)}} \!\cdot
\!\left(\frac{\Delta a}{a^{(k)}} -1 \right) \nonumber \\
\relax&  - \frac{2 \Re\{\Delta \mathbf{w}^H
(\mathbf{Q}_m
\!+\!(|d_m|^2 /\sigma_\nu^2)\mathbf{R}_m)
\mathbf{w}^{(k)}\} }{a^{(k)}} \leq 0. \label{ConvexifiedConstraint}
  \end{align}
In \cite{MultiUserTwoWay}, a similar approximation has been used for
 max-min fair beamforming in bi-directional relay networks. The
approximation made in inequality (\ref{ConvexifiedConstraint}) is however tighter than that of \cite{MultiUserTwoWay} in the
sense that  unlike the approach of \cite{MultiUserTwoWay}, the convex
quadratic-over-linear
terms are not linearized. In \cite{30}, a linearization approach for max-min
fair beamforming in the context of cognitive radio networks has been proposed.
The latter approach assumes a centralized system and therefore does not
consider the problem of power allocation which arises in our distributed
beamforming application. Power allocation and beamforming optimization have been
addressed in \cite{Yong}, where a CCCP algorithm has been proposed to minimize
the transmit power in a cooperative relay network. An iterative algorithm
for power minimization in a Rank-2-AFMS has been proposed in our accompanying conference paper \cite{Eusipco}.
In the referenced work, however, direct source-destination channels and power allocation
have not been regarded.

\textcolor{black}{
Comparing inequalities (\ref{SNR-Constraint}) and (\ref{ConvexifiedConstraint}), we find that
\begin{equation}
{\lambda}_m( \mathbf{w}^{(k)} +\Delta \mathbf{w}, a^{(k)} +\Delta a, t^{(k)} + \Delta t)  \leq \bar{\lambda}_m^{(k)}(\Delta \mathbf{w},\Delta a, \Delta t),
\label{InnerApproxProp}
\end{equation}
which is a consequence of the linearization of the concave part of ${\lambda}_m( \mathbf{w},  a, t)$, see Section 3.1.3 in \cite{Boyd}.
Therefore, the convex problem
\begin{align}
 \min_{\Delta \mathbf{w},\Delta a, \Delta t} & t^{(k)} + \Delta t
\nonumber \\ \relax   {\rm
s.t.}\quad &  \bar{\lambda}_m^{(k)}(\Delta \mathbf{w},\Delta a, \Delta t) \leq
0~~\forall
m\in\{1,\dots,M\},
\nonumber \\ \relax
&  t^{(k)} +\Delta t > 0, \nonumber \\ \relax
& ( \mathbf{w}^{(k)}+\Delta \mathbf{w},a^{(k)}+\Delta a) \in \Omega,
\label{Opt:SNR_MaximizationConvexified}
\end{align}
represents an inner approximation of $\mathcal{M}$.
}
\newtheorem{bearnie}[Satz]{Corollary}
\begin{bearnie}
Let $(\mathbf{w}^{(k)}, a^{(k)}, t^{(k)})$ be feasible for $\mathcal{M}$. The updated variables
$(\mathbf{w}^{(k+1)}, a^{(k+1)}, t^{(k+1)})$, obtained from a solution
$(\Delta \mathbf{w}^\star,\Delta a^\star,\Delta t^\star)$ to problem
(\ref{Opt:SNR_MaximizationConvexified}), are feasible for $\mathcal{M}$ and
$t^{(k+1)} \leq
t^{(k)} $.
\end{bearnie}
\begin{proof}
$\Delta t^\star \leq 0$ and consequently $t^{(k+1)} \leq
t^{(k)} $ hold true as $(\Delta \mathbf{w} = \mathbf{0},\Delta a=0,\Delta t=0)$  is feasible for problem (\ref{Opt:SNR_MaximizationConvexified}) since $(\mathbf{w}^{(k)}, a^{(k)}, t^{(k)})$ is feasible. From inequality (\ref{InnerApproxProp}) follows  ${\lambda}_m( \mathbf{w}^{(k+1)}, a^{(k+1)}, t^{(k+1)} )\leq 0$ and, consequently, the SNR constraints are fulfilled.  As $t^{(k+1)} >0$ and
$ (\mathbf{w}^{(k+1)},a^{(k+1)}) \in \Omega$ are ensured according to problem (\ref{Opt:SNR_MaximizationConvexified}), the feasibility of $(\mathbf{w}^{(k+1)}, a^{(k+1)}, t^{(k+1)})$ for $\mathcal{M}$ is guaranteed.
\end{proof}

The property $1/t^{(k+1)} \geq
1/t^{(k)} $ ensures that the minimum SNR increases or
remains unchanged in each iteration.
Repeatedly solving problem (\ref{Opt:SNR_MaximizationConvexified}) for $k=0,1,2,\dots $
creates a monotonically non-decreasing sequence of minimum SNR values
$1/t^{(0)}
\leq 1/t^{(1)} \leq 1/t^{(2)} \dots $ with feasible weight vectors  and a feasible
transmit power. \textcolor{black}{Interestingly, we can show that the latter update iteration does not exhibit divergence or oscillation. \textcolor{black}{Moreover, the sequence $(\mathbf{w}^{(k)}, a^{(k)}, t^{(k)})$ converges globally, i.e., for every feasible start point, to a stationary point of $\mathcal{M}$ that satisfies the
Karush-Kuhn-Tucker (KKT) conditions. \footnote{For points where constraint qualification holds, the KKT conditions are necessary for local optimality.
Note that for convex optimization problems, the KKT conditions are sufficient to guarantee global optimality if Slater's constraint qualification is satisfied \cite{Boyd}.} For non-convex optimization problems, a KKT point can be a local optimum, a global optimum, a maximum, or a saddle point \cite{Boyd}. It is easy to verify that for $\mathcal{M}$, a KKT point cannot be a maximum: For any point $( \mathbf{w}, a,  t)$, reducing the source power by choosing a larger power scaling factor $a$ leads to a larger target function value $t$. We remark that it is in general NP-hard to prove the local optimality of a KKT point \cite{localopt}. Even though a proof of optimality of a KKT point is difficult  to derive, the search for KKT points has been found useful to solve non-convex optimization problems \cite{Vanderbei}, \cite{Gill}.
}
\newtheorem{bearnie2}[Satz]{Theorem}
\begin{bearnie2}
Let us assume that the power factor $a$ is bounded by $a_{\max} \geq a$. Then, for any feasible initial point $(\mathbf{w}^{(0)}, a^{(0)}, t^{(0)})$, the sequence $(\mathbf{w}^{(k)}, a^{(k)}, t^{(k)})$ converges to a stationary point.
\end{bearnie2}
\begin{proof}
According to Theorem 10 of \cite{CCCP},  the sequence $(\mathbf{w}^{(k)}, a^{(k)}, t^{(k)})$ {\it globally} converges if the mapping $(\mathbf{w}^{(k)}, a^{(k)}, t^{(k)})\rightarrow (\mathbf{w}^{(k+1)}, a^{(k+1)}, t^{(k+1)})$ of the definitions (\ref{map1}) and (\ref{map2}) is uniformly compact. This is the case if the feasible set of $\mathcal{M}$ is compact \cite{CCCP}. We show that the variables lie inside a compact set, i.e., a set that is closed and bounded: Using our assumption and the constraint on the source power (\ref{constraint:Source_Power}), $a$ lies in the interval  $ 2/P_{S,\max} \leq a \leq a_{\max} $ which is closed and bounded. Moreover, due to the power constraints (\ref{constraint:Ind_Power}) - (\ref{Opt:Powerconstraints2}), it is clear that the set of feasible weight vectors $\mathbf{w}$ is bounded and closed. $t$ lies in the  closed and bounded interval $t_{\rm opt} \leq t \leq t^{(0)} $, where $t_{\rm opt}$ is the (unknown) optimum value.
\end{proof}
For the latter proof, we have assumed that $  a \leq a_{\max} $ which implies that the source power during the first two time slots does not vanish.
 Adding $  a \leq a_{\max} $ as an additional constraint to $\mathcal{M}$ is not a critical requirement provided that  $a_{\max}$ is chosen sufficiently large. Then, the approximated problem (\ref{Opt:SNR_MaximizationConvexified})
becomes
\begin{align}
 \min_{\Delta \mathbf{w},\Delta a, \Delta t} & t^{(k)} + \Delta t
\nonumber \\ \relax   {\rm
s.t.}\quad &  \bar{\lambda}_m^{(k)}(\Delta \mathbf{w},\Delta a, \Delta t) \leq
0~~\forall
m\in\{1,\dots,M\},
\nonumber \\ \relax
&  a^{(k)} +\Delta a \leq a_{\max}, \nonumber \\ \relax
&  t^{(k)} +\Delta t > 0, \nonumber \\ \relax
& ( \mathbf{w}^{(k)}+\Delta \mathbf{w},a^{(k)}+\Delta a) \in \Omega.
\label{Opt:SNR_MaximizationConvexified2}
\end{align}}

The globally convergent Max-Min-CCCP algorithm outlined in Algorithm 1 starts at a random point and
iterates until
the
relative progress $\rho \triangleq |t^{(k)} - t^{(k+1)}|/t^{(k)}$ falls below
the
threshold
value
$\varepsilon.$
\\
\begin{algorithm}
\DontPrintSemicolon
\Begin{
Set $k:=0$. Create and scale random $\mathbf{w}^{(0)} \neq
\mathbf{0}_{2R+2}$ and $a^{(0)} \neq 0 $ such that $ ( \mathbf{w}^{(0)},a^{(0)}) \in \Omega$. Compute $t^{(0)} = 1/(\min {\rm SNR}_m)$.  Set $\rho >
\varepsilon$. \\
  \While{$\rho > \varepsilon$}{$k:=k+1.$ \\Compute $(\Delta
\mathbf{w}^\star,\Delta
a^\star,\Delta t^\star)$ by solving problem
(\ref{Opt:SNR_MaximizationConvexified2}), using
$(\mathbf{w}^{(k)},a^{(k)},t^{(k)})$, update $\mathbf{w}^{(k+1)} =
\mathbf{w}^{(k)}+ \Delta
\mathbf{w}^\star$, $a^{(k+1)} = a^{(k)}+ \Delta
a^\star$ and $t^{(k+1)} = t^{(k)}+ \Delta
t^\star$.\\
$\rho:= |t^k - t^{k+1}|/t^k.
$}}
\Return $\mathbf{w}^{(k)},a^{(k)},t^{(k)}$
  \caption{Max-Min-CCCP} \label{Alg2}
\end{algorithm}

As the proposed algorithm is based on the linearization of the
concave (negative convex) functions of a DC program, it belongs to
the class of CCCP algorithms \cite{CCCP}.

Note that in our algorithm, the problem
(\ref{Opt:SNR_MaximizationConvexified2}) is solved exactly. To reduce the
computational cost, it is possible to use an inaccurate solution \cite{ERt}. A
detailed description of an implementation is
beyond the scope of this work. In our simulations, the subproblems that arise in
every iteration of our proposed algorithm are solved exactly.

\section{Simulation Results}
To evaluate the performance of the proposed scheme under realistic conditions, we consider a relay
network, where the
coefficients of the source-relay, relay-destination and source-destination channels
model an urban micro scenario \cite{Baum}.
The system parameters are chosen according to the {L}ong
{T}erm
{E}volution (LTE) standard for mobile communication \cite{4G}. The system
is operated at
a carrier
frequency
of $1800$ MHz and we choose $T_S\! =\!66.7\mu {\rm sec} $ as
the duration of one time slot. The bandwidth is given by $1/T_S$  which
corresponds to the
bandwidth of a subcarrier in a multi-carrier LTE system. We further assume frequency flat fading
channels. We create the channel coefficients such that there is no
shadow fading from the
source to the relays but from the source and
relays to the destinations. The noise power is set to
$\sigma_\nu^2\!= \sigma_\eta^2=\!-132$ dBm. The maximum transmit power
values are chosen
as
$P_{T,\max}$, $ P_{S,\max} = P_{T,\max}/2$, $ P_{R,\max} = P_{T,\max}/3$, and $
p_{r,\max} = P_{T,\max}/15$, respectively.

In the network, $R=10$ relays are  placed at
a distance of 250 meters around the source at
equidistant angles. The destinations are randomly
distributed in between 600 and 800 meters around the source, see Fig.
\ref{Setup}. The source, the relays, and the destinations are placed at a height
of 10, 5, and 1.5 meters, respectively.

We investigate the performance of the following transmission schemes, scenarios,
and algorithms: The Rank-2-AFMS combined with the
Max-Min-CCCP algorithm (R2-Max-Min-CCCP) and the SDR2D algorithm (R2-SDR2D).
Furthermore, we consider the Rank-1-AFMS
combined with the CCCP algorithm (R1-Max-Min-CCCP) and the SDR2D algorithm
(R1-SDR2D). Moreover, we investigate also direct source-destination
(DSD)
communication, where the relays are not involved.
All results are compared with
the theoretical upper bound (SDR2D-UB) obtained by the SDR2D algorithm, see
Subsection~\ref{sR2}.

For the SDP-feasibility problems, which are solved in the  SDR2D algorithm to find a solution of problem (\ref{Opt:MaxMinSDRAFMS}), and to solve
the subproblems
(\ref{Opt:SNR_MaximizationConvexified2}) of the Max-Min-CCCP algorithm, we have
utilized the cvx interface for convex programming in a Matlab environment
\cite{cvx}. As the solver, we have chosen Mosek 7.0.0.103 under the default precision \cite{Mosek}. For the SDR2D algorithm, we search
over 200 grid points to determine the power scaling factor $a$ and select $\varepsilon = 10^{-2}$ as the precision
in the bisection search to determine the optimum $t$. We chose $\varepsilon$ also as the threshold value
for the relative progress of the Max-Min-CCCP algorithm. We select $a_{\max}=2\cdot 10^{6}/P_{S,\max}$ in the problem (\ref{Opt:SNR_MaximizationConvexified2})
that is solved in every iteration of the Max-Min-CCCP algorithm.

 In the case that the solution matrix to the optimization problem
(\ref{Opt:MaxMinSDRAFMS}) obtained by the SDR2D algorithm is not feasible for
$\mathcal{M}$, we apply the Gaussian randomization procedure
of \cite{Xin}.
For the Rank-1-AFMS, we select the
best out of
200 random vectors, where each vector is properly scaled to meet the power
constraints. For the Rank-2-AFMS, we
create 200 random vector pairs. Each pair is properly scaled by solving a linear
program using the function \texttt{linprog.m} of the Matlab optimization
toolbox \cite{Matlab}.

\begin{figure}
	\psfrag{Transmitter}{source}
	\psfrag{Relays}{relays}
	\psfrag{Receivers}{destinations}
	\psfrag{destinations}{\bf Destinations}
	\psfrag{f}{\small $\bff$}
	\psfrag{g1}{\small $\bg_1$}
	\psfrag{g2}{\small $\bg_M$}
	\psfrag{g_N}{\small $\bg_N$}
\psfrag{r1}{\small $r_1$}
	\psfrag{r2}{\small $r_2$}
  \centerline{\epsfig{figure = 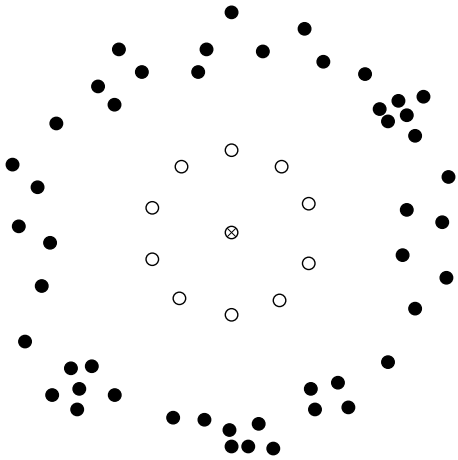,width=6.0cm,height = 6cm}}
\caption{Setup of the network; $\otimes$: source, $\circ $:
relay, $\bullet$: destination.} \label {Setup}
\end{figure}

In our first example, we examine the performance of all setups in
terms of the average minimum achieved rate by solving $\mathcal{M}$. The minimum
achieved rate is given by $(1/2){\rm log}_2(1+ {\rm SNR})$ (the prefactor
$1/2$ takes
into account the time slots per communicated symbol) for communications with
relays and by ${\rm log}_2(1+ {\rm SNR})$ for DSD, as one symbol per time
slot can be communicated.

Fig.~\ref{Fig:FirstExamplePlot1} depicts the average minimum rate versus
$P_{T,\max} $ in the case that the direct source-destination channels are
exploited for $M=100$ destinations. In our simulation results depicted in
Fig.~\ref{Fig:FirstExamplePlot1b} it is assumed that no direct
source-destination
channels exist.

Both figures demonstrate that R2-Max-Min-CCCP achieves near-optimum
performance close to the theoretical upper bound. The Rank-1-AFMS is clearly
outperformed
by the Rank-2-AFMS. Comparing Fig.~\ref{Fig:FirstExamplePlot1} with
Fig.~\ref{Fig:FirstExamplePlot1b}, we observe that the performance improves
significantly if the source-destination channels are exploited.

Fig.~\ref{Fig:FirstExamplePlot2} depicts the average minimum rate versus the
number of destinations $M$
in the case that the direct source-destination channels are
exploited for $P_{T,\max} =$5dBm. In Fig.~\ref{Fig:FirstExamplePlot2b} the same
setups are considered, assuming that
no direct source-destination
channels exist.

Both Fig.~\ref{Fig:FirstExamplePlot2} and Fig.~\ref{Fig:FirstExamplePlot2b}
demonstrate that R2-Max-Min-CCCP achieves
near-optimum
performance close to the upper bound. For small $M$, the SDR2D algorithm
performs slightly better than the Max-Min-CCCP algorithm due to the fact, that
in most cases,
the
respective solution matrix has a rank smaller or equal to two. As illustrated
in Table~\ref{table1}, the rank of the SDR solution matrix increases with
increasing $M$. This leads to suboptimal solutions generated by the
randomization technique.

\begin{table}
\begin{center}
\begin{tabular}{| c || c | c | c| c | c|}
\hline
$M$ & 10& 40 &70 &100& 130 \\
Average $\mathcal{R}(\mathbf{X}^\star_1)$  &  1.8 & 2.5 & 3.0 & 3.3 & 3.4 \\
\hline
\end{tabular}
\end{center}
\caption{Average rank $\mathcal{R}(\mathbf{X}^\star_1)$ versus number of
destinations
$M$.} \label{table1}
\end{table}

In our second example, we examine the computational aspects of the considered
algorithms for $M=100$ destinations and $P_{T,\max} =$20dBm. Table \ref{table2} shows the average
runtime in seconds
of the
different algorithms.  We observe that the Max-Min-CCCP algorithm is more than
ten times faster than the SDR2D algorithm. Fig.~\ref{Fig:FirstExamplePlot3}
depicts the number of iterations of the Max-Min-CCCP algorithm for the
Rank-2-AFMS versus the minimum SNR.  \textcolor{black}{ In each simulation run, we have initialized the R2-Max-Min-CCCP method with ten different starting points. In Fig.~\ref{Fig:FirstExamplePlot3}, the highest (R2-Max-Min-CCCP H) and the lowest minimum SNR (R2-Max-Min-CCCP L) achieved within these ten optimization runs of the Max-Min-CCCP algorithm are depicted for each iteration.  As a comparison to the CCCP algorithm, the performance of the
SDR-based approaches is depicted. To analyze the influence of the randomization procedure on the performance of the SDR2D algorithm, we increase the number of random vectors from 200 to 1000 for the R1-SDR2D method (R1-SDR2D 1000) and for the R2-SDR2D method (R2-SDR2D 1000). As expected, the increased number of random vectors leads to a higher minimum SNR, however, the R2-Max-Min-CCCP still outperforms the SDR-based designs. Furthermore, we observe that the main progress of the Max-Min-CCCP algorithm is achieved within the first three iterations. In the following iterations,  the gain that is obtained is less than 1dB. Since  the Max-Min-CCCP algorithm exhibits excellent performance within  a few iterations it is suitable for real time applications. The difference between the highest minimum SNR and the lowest minimum SNR achieved with ten different starting points decreases with the number of iterations. For more than nine iterations, the difference is less than 0.5 dB. We remark from this observation that the performance can be improved if the R2-Max-Min-CCCP is initialized with different starting points.}

\begin{table}
\begin{center}
\textcolor{black}{\begin{tabular}{| c || c |  c| c | c|}
\hline
\!\! \!\!\textbf{Algorithm} \!\!\!\! & \!\!\!\! SDR2D,
R1-Rand\!\!\!\! &\!\! \!\! \!\!SDR2D, R2-Rand\!\!\!\! \!\!& \!\! \!\!R1-CCCP \!\!\!\! &
\!\!\!\! R2-CCCP \!\!\!\! \\
\!\!\textbf{Runtime} \!\! &  $1.0 \cdot 10^3$ &$ 1.1 \cdot 10^3$ & 81 & 83 \\
\hline
\end{tabular}}
\end{center}
\caption{Average runtime in seconds. R1-Rand:
Rank-one randomization, R2-Rand: Rank-two randomization, R1-CCCP:
R1-Max-Min-CCCP,  R2-CCCP: R2-Max-Min-CCCP.} \label{table2}
\end{table}

\begin{figure}
\psfrag{R2AFMF+CCCP}{\scriptsize{R2-Max-Min-CCCP}}
\psfrag{R1AFMF+CCCP}{\scriptsize{R1-Max-Min-CCCP}}
\psfrag{R2AFMF+SDR2D+LB}{\scriptsize{SDR2D-UB}}
\psfrag{R2AFMF+SDR2D}{\scriptsize{R2-SDR2D}}
\psfrag{R1AFMF+SDR2D}{\scriptsize{R1-SDR2D}}
\psfrag{DSD}{\scriptsize{DSD}}
\psfrag{XXXXXXXXXMinimal Rate (bits/time slots/Hz)}{\small{~~Minimum rate
(bits/time
slot/Hz)}}
\psfrag{Overall Power PTmax}{\small{\!\!\!\!\! Total power $P_{T,\max} $ (dBm)}}
\centering
\centerline{\epsfig{figure=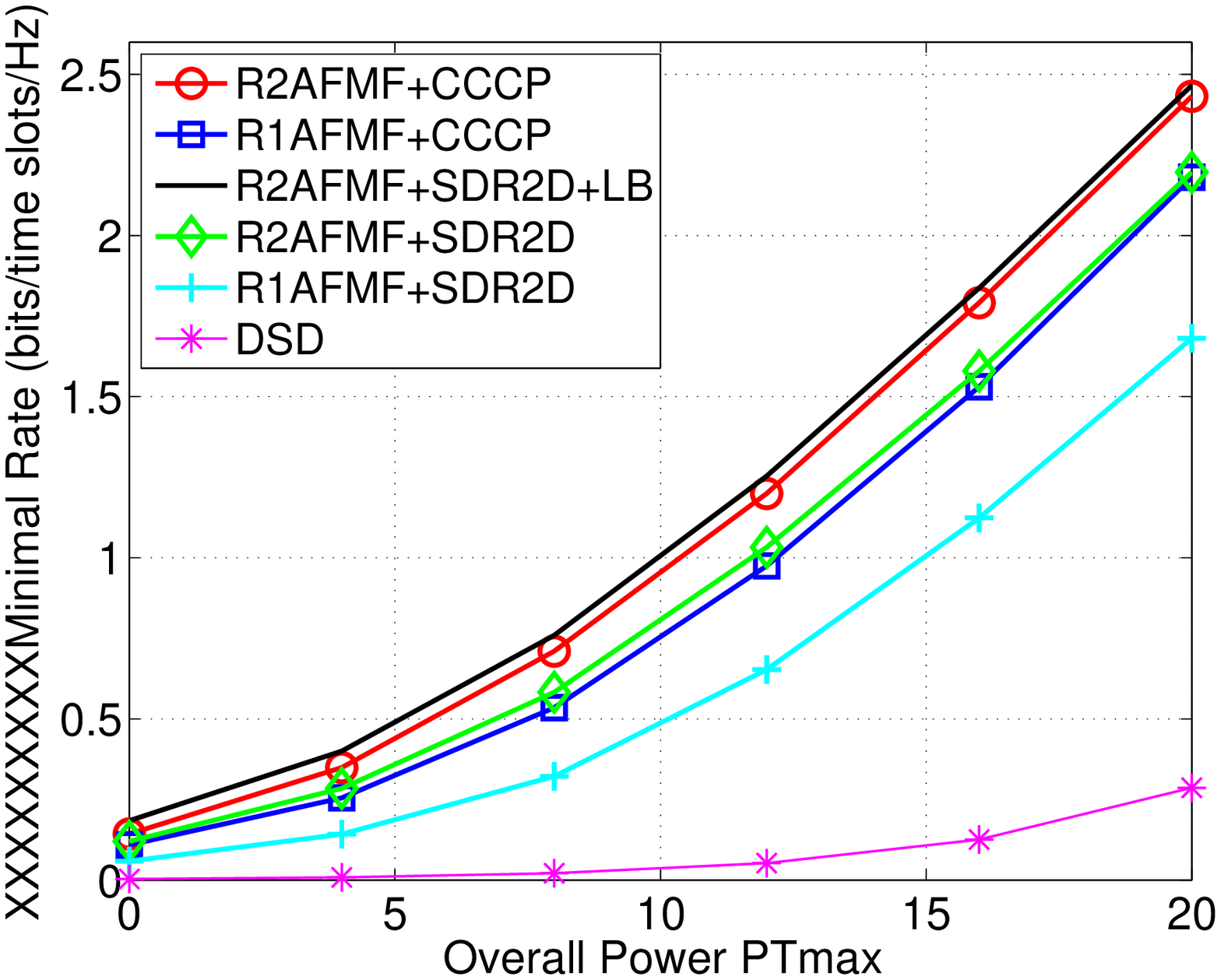,scale=0.42}}
\caption{Minimum rate versus total power $P_{T,\max} $, first example.}
\label{Fig:FirstExamplePlot1}
\end{figure}

\begin{figure}
\psfrag{R2AFMF+CCCP}{\scriptsize{R2-Max-Min-CCCP}}
\psfrag{R1AFMF+CCCP}{\scriptsize{R1-Max-Min-CCCP}}
\psfrag{R2AFMF+SDR2D+LB}{\scriptsize{SDR2D-UB}}
\psfrag{R2AFMF+SDR2D}{\scriptsize{R2-SDR2D}}
\psfrag{R1AFMF+SDR2D}{\scriptsize{R1-SDR2D}}
\psfrag{direct path only}{\scriptsize{Direct communications without relays}}
\psfrag{XXXXXXXXXMinimal Rate (bits/time slots/Hz)}{\small{~~Minimum rate
(bits/time
slot/Hz)}}
\psfrag{Overall Power PTmax}{\small{\!\!\!\!\! Total power $P_{T,\max} $ (dBm)}}
\centering
\centerline{\epsfig{figure=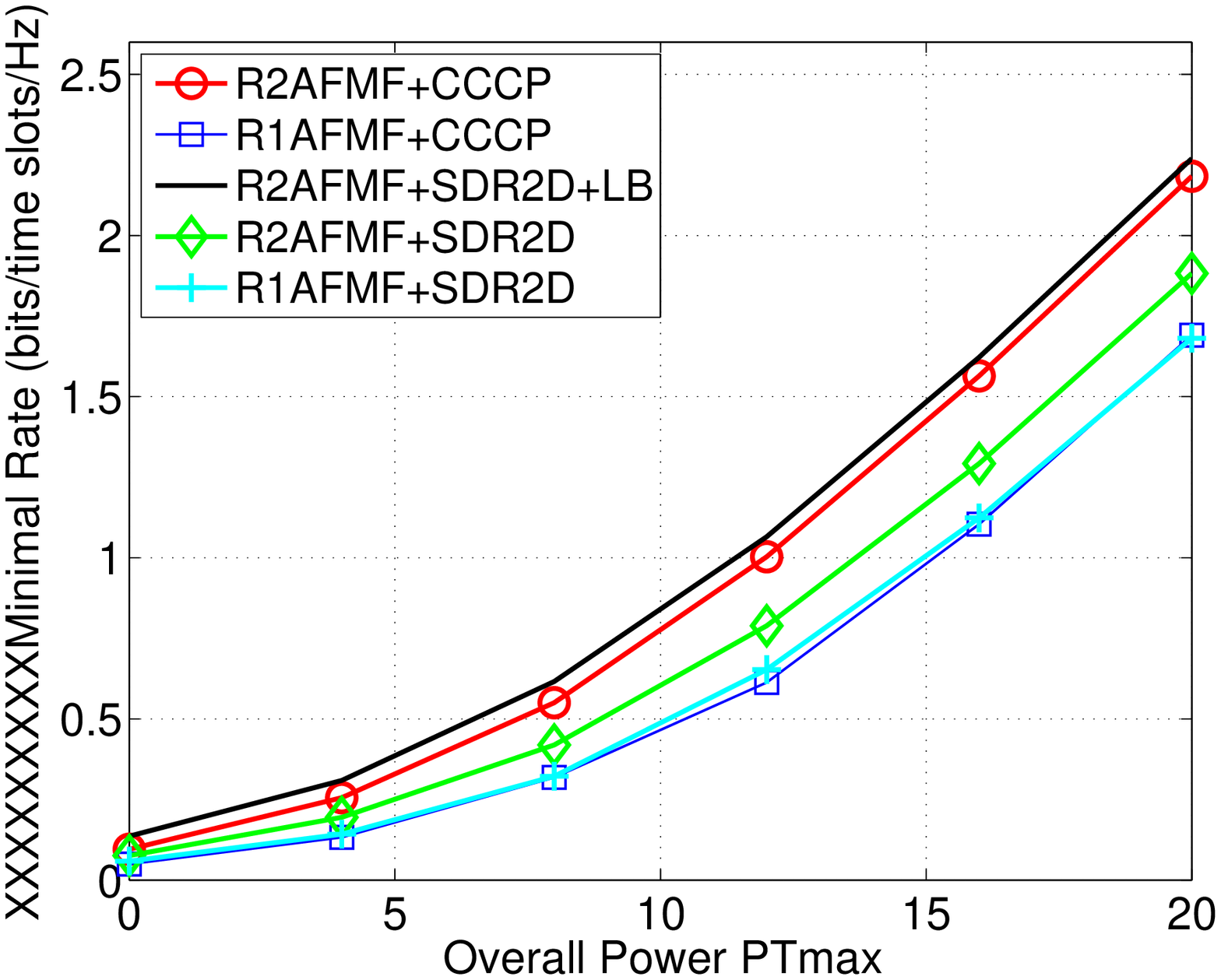,scale=0.42}}
\caption{Minimum rate versus total power $P_{T,\max} $, first example.}
\label{Fig:FirstExamplePlot1b}
\end{figure}

\begin{figure}
\psfrag{R2AFMF+CCCP}{\scriptsize{R2-Max-Min-CCCP}}
\psfrag{R1AFMF+CCCP}{\scriptsize{R1-Max-Min-CCCP}}
\psfrag{R2AFMF+SDR2D+LB}{\scriptsize{SDR2D-UB}}
\psfrag{R2AFMF+SDR2D}{\scriptsize{R2-SDR2D}}
\psfrag{R1AFMF+SDR2D}{\scriptsize{R1-SDR2D}}
\psfrag{DP}{\scriptsize{DSD}}
\psfrag{direct path only}{\scriptsize{Direct communications without relays}}
\psfrag{XXXXXXXXXMinimal Rate (bits/time slots/Hz)}{\small{~~Minimum rate
(bits/time
slot/Hz)}}
\psfrag{Number of Users M}{\small{\!\!\!\!\!\!\!\!\!\!\!\!\!\!\!Number of
destinations $M$}}
\centering
\centerline{\epsfig{figure=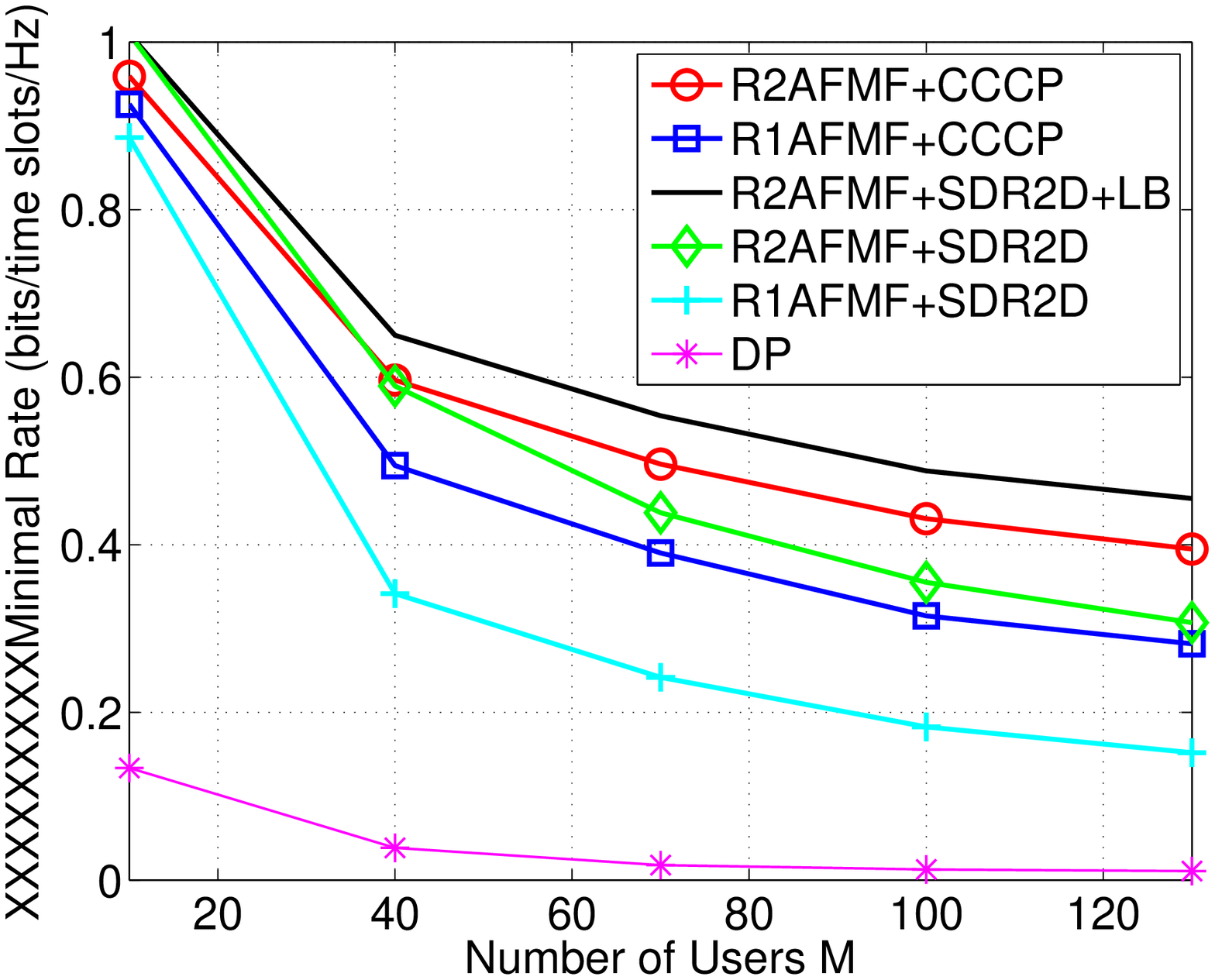,scale=0.42}}
\caption{Minimum rate versus number of destinations $M$, first example.}
\label{Fig:FirstExamplePlot2}
\end{figure}

\begin{figure}
\psfrag{R2AFMF+CCCP}{\scriptsize{R2-Max-Min-CCCP}}
\psfrag{R1AFMF+CCCP}{\scriptsize{R1-Max-Min-CCCP}}
\psfrag{R2AFMF+SDR2D+LB}{\scriptsize{SDR2D-UB}}
\psfrag{R2AFMF+SDR2D}{\scriptsize{R2-SDR2D}}
\psfrag{R1AFMF+SDR2D}{\scriptsize{R1-SDR2D}}
\psfrag{direct path only}{\scriptsize{Direct communications without relays}}
\psfrag{XXXXXXXXXMinimal Rate (bits/time slots/Hz)}{\small{~~Minimum rate
(bits/time
slot/Hz)}}
\psfrag{Number of Users M}{\small{\!\!\!\!\!\!\!\!\!\!\!\!\!\!\!Number of
destinations $M$}}
\centering
\centerline{\epsfig{figure=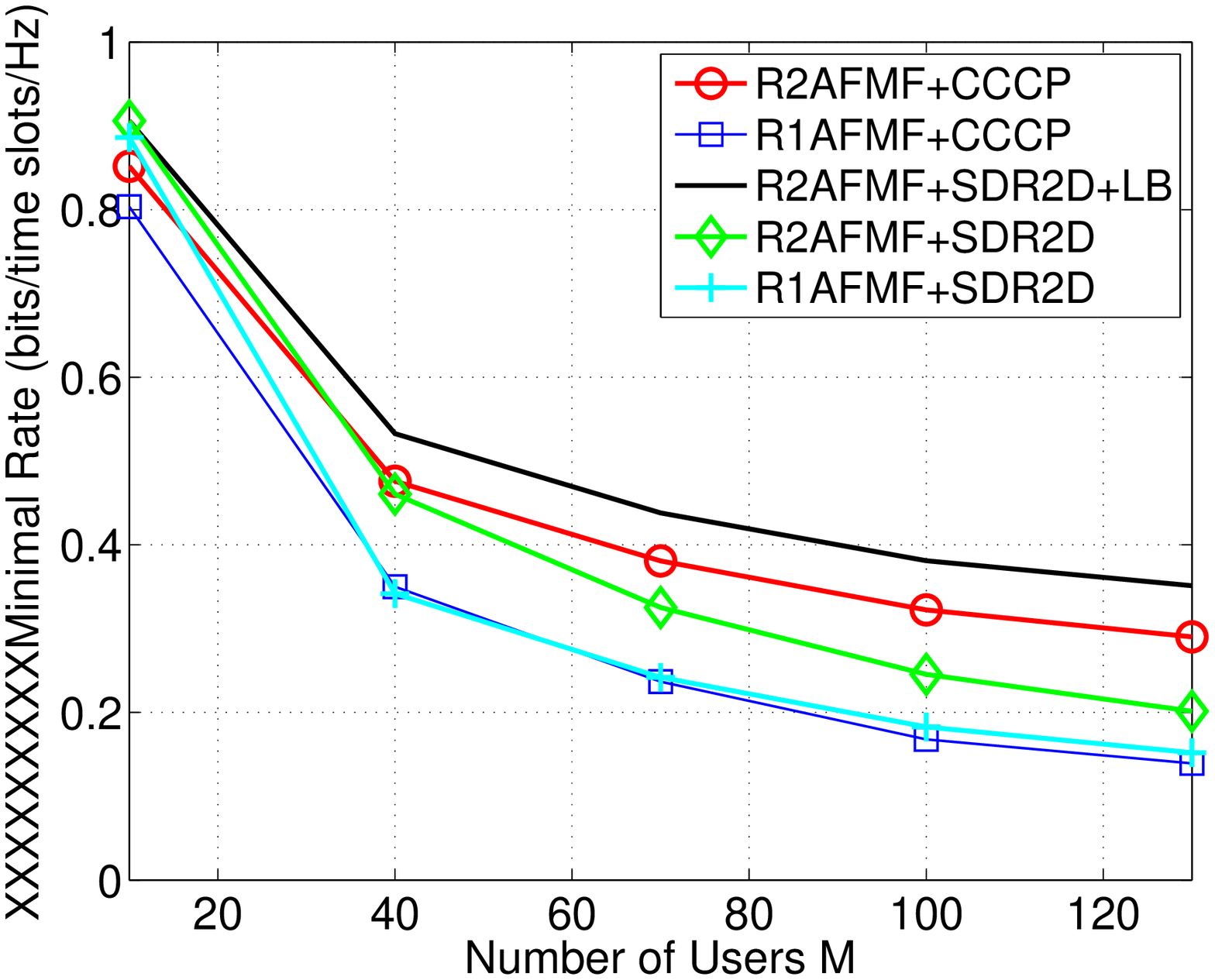,scale=0.42}}
\caption{Minimum rate versus number of destinations $M$, first example.}
\label{Fig:FirstExamplePlot2b}
\end{figure}

\begin{figure}
\psfrag{R2AFMF+CCCP}{\scriptsize{R2-Max-Min-CCCP}}
\psfrag{R1AFMF+CCCP}{\scriptsize{R1-Max-Min-CCCP}}
\psfrag{R2AFMF+SDR2D+LB}{\scriptsize{SDR2D-UB}}
\psfrag{R2AFMF+SDR2D}{\scriptsize{R2-SDR2D}}
\psfrag{R1AFMF+SDR2D}{\scriptsize{R1-SDR2D 1000}}

\psfrag{R2AFMF+CCCP best}{\scriptsize{R2-Max-Min-CCCP H}}
\psfrag{R2AFMF+CCCP worst}{\scriptsize{R1-Max-Min-CCCP L}}
\psfrag{R2AFMF+SDR2D+LB}{\scriptsize{SDR2D-UB}}
\psfrag{R2AFMF+SDR2D 1000}{\scriptsize{R2-SDR2D 1000}}
\psfrag{R1AFMF+SDR2D 1000}{\scriptsize{R1-SDR2D}}

\psfrag{direct path only}{\scriptsize{Direct communications without relays}}
\psfrag{XXXXXXXXXMinimal Rate (bits/time slots/Hz)}{\small{ ~~~~~~~
Minimum SNR (dB)}}
\psfrag{Overall Power PTmax}{\small{Number of iterations}}
\centering
\centerline{\epsfig{figure=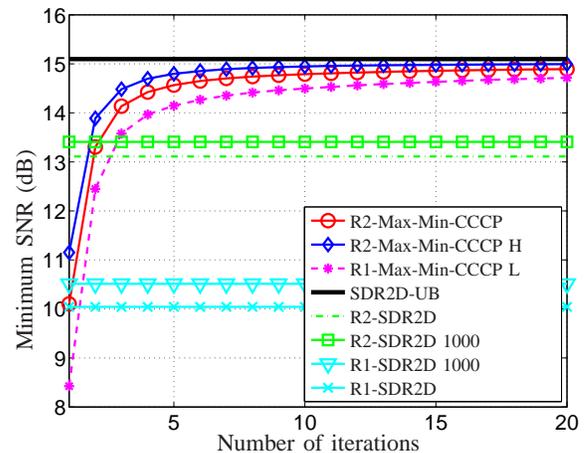,scale=0.42}}
\caption{Minimum SNR versus number of iterations, second example.}
\label{Fig:FirstExamplePlot3}
\end{figure}

\section{Conclusion}
In this paper, a novel Rank-2-AFMS for single-group multicasting
has been proposed. Our scheme generalizes the rank-one multicasting
scheme of the literature to a rank-two multicasting scheme and
allows to incorporate the direct channel from the source to the destinations in
the detection. To select
the proper source power and to adjust the relay weights we
propose an iterative algorithm to maximize the lowest SNR at the destinations.
The simulation results demonstrate
that the proposed Rank-2-AFMS combined with the proposed
iterative algorithm yields  a performance close to the theoretical
upper bound.

\section{Acknowledgement}
This work was supported by the Seventh Framework Programme for Research of the European Commission under grant number ADEL-619647.

\end{document}